\documentclass[12pt]{amsart}

\usepackage{microtype}%if unwanted, comment out or use option "draft"

\usepackage{amsmath}
\usepackage{hyperref}
\usepackage{amsmath}
\usepackage{amssymb}
\usepackage{tikz}
\usepackage{color}
\usepackage{mathpartir}

\usepackage[centertags]{amsmath}

\usepackage[mathscr]{euscript}

\usepackage{amsthm}

\usepackage{tikz}

\usepackage{amsmath}

\usepackage{amssymb}

\usepackage{empheq}

\usepackage{newlfont}

\usepackage{amsfonts}

\usepackage{empheq}

\usepackage{newlfont}

\usepackage{graphicx}

\usepackage{color}

\usepackage{epsfig}

\usepackage{bm}

\usepackage{prooftree}

\title{The $n$-dimensional Propositional Calculus}

\newcommand{\e}{\mathsf e}

\usepackage{microtype}%if unwanted, comment out or use option "draft"
\usepackage{amsmath}
\usepackage{amsmath}
\usepackage{amssymb}
\usepackage{tikz}
\usepackage{color}
\newcommand{\cc}{\mathsf c}
\newcommand{\ch}{\mathsf h}
\newcommand{\dd}{\mathsf d}

\newcommand\TTTT{%
 \textsf{T\kern-0.15em\raisebox{-0.55ex}T\kern-0.15emT\kern-0.15em\raisebox{-0.55ex}2}%
}
\providecommand{\U}[1]{\protect\rule{.1in}{.1in}}
\newtheorem{theorem}{Theorem}
\theoremstyle{plain}

\newtheorem{corollary}{Corollary}

\newtheorem{definition}{Definition}
\newtheorem{example}{Example}

\newtheorem{lemma}{Lemma}

\newtheorem{proposition}{Proposition}

\numberwithin{equation}{section}

\title[From Boolean products to semiring products]{Boolean-like algebras of finite dimension: From Boolean products to semiring products   %\footnote{This work was partially supported by someone.}
}
\author[A. Bucciarelli]{A. Bucciarelli}\address{A. Bucciarelli, Universit\'e de Paris}\email{buccia@irif.fr}
\author[A. Ledda]{A. Ledda}\address{A. Ledda, F. Paoli, Universit\`a di Cagliari}\email{antonio.ledda@unica.it,\, paoli@unica.it}
\author[F. Paoli]{F. Paoli}%\address{Universit\`a di Cagliari}\email{paoli@unica.it}
\author[A. Salibra]{A. Salibra}\address{A. Salibra, Universit\'e de Paris}\email{salibra@unive.it}
\date{\today}
\thanks{{\bf Corresponding author}: {\bf Francesco Paoli}, {\ttfamily paoli@unica.it}
}
\begin{document}

\maketitle
\begin{abstract}
We continue the investigation, initiated in \cite{SBLP}, of Boolean-like algebras of dimension $n$ ($n\mathrm{BA}$s), algebras having $n$ constants $\e_1,\dots,\e_n$, and an $(n+1)$-ary operation $q$ (a ``generalised if-then-else'') that induces a decomposition of the algebra into $n$ factors through the so-called $n$-central elements. Varieties of $n\mathrm{BA}$s  share many remarkable properties with the variety of Boolean algebras and with primal varieties. Putting to good use the concept of a central element, we extend the Boolean power construction to that of a semiring power and we prove two representation theorems: (i) Any pure $n\mathrm{BA}$ is isomorphic to the algebra of $n$-central elements of a Boolean vector space; (ii) Any member of a variety of $n$BAs with one generator is isomorphic to a Boolean power of this generator. This yields a new proof of Foster's theorem on primal varieties.
%The $n\mathrm{BA}$s provide the algebraic framework for generalising the classical propositional calculus to the case of $n$ - perfectly symmetric - truth-values. Every finite-valued tabular logic can be embedded into such an $n$-valued propositional logic, $n\mathrm{CL}$, and this embedding preserves validity. We define a confluent and terminating first-order rewriting system for deciding validity in $n\mathrm{CL}$, and, via the embeddings, in all the finite tabular logics. %Deciding validity by computation, as opposed to truth tables, bears resemblance to multivalued decision diagrams.

\end{abstract}

\section{Introduction}

Janusz Czelakowski began his scientific career with a series of investigations on partial Boolean algebras, in the context of a semantical analysis of quantum logic (see e.g. \cite{Czel1, Czel2, Czel3}. A different variation on the concept of a Boolean algebra will be at the centre of this paper: the notion of a \emph{Boolean-like algebra of finite dimension}.

In a number of papers \cite{first, SLP2017, SLP18, SBLP} we attempted to combine an algebraic approach to if-then-else statements in programming languages with a general theory of direct decompositions in universal algebra. A \emph{Church algebra of dimension} $2$ is an algebra $\mathbf{A}$ whose type contains a ternary term $q$ (having the if-then-else connective as an intended interpretation) and two constants $0$ and $1$ (representing falsity and truth, respectively); it has the property that for every $a,b\in A$, $q^{\mathbf{A}}( 1^{\mathbf{A}},a,b) =a$ and $q^{\mathbf{A}}( 0^{\mathbf{A}},a,b) =b$. An element $e$ of a Church algebra $\mathbf{A}$ of dimension $2$ is called $2$\emph{-central} if $\mathbf{A}$ can be decomposed as the product $\mathbf{A}/\theta ( e,0) \times \mathbf{A}/\theta ( e,1) $, where $\theta ( e,0) $ ($\theta ( e,1) $) is the smallest congruence on $\mathbf{A}$ that collapses $e$ and $0$ ($e$ and $1$). \emph{Boolean-like algebras of dimension }$2$, investigated in \cite{first} under the name of \emph{Boolean-like algebras}, are Church algebras of dimension $2$ where every element is $2$-central -- the chosen denomination being justified by the fact that the variety of all such algebras in the language $(q,0,1)$ is term-equivalent to the variety of Boolean algebras. 

It is all too natural to generalise this viewpoint to cover the case of decompositions of algebras into finitely many factors. In \cite{SBLP}, we considered algebras $\mathbf{A}$ having $n$ designated elements $\e_1,\dots,\e_n$ ($n\geq 2$) and an $n+1$-ary operation $q$ (a sort of "generalised if-then-else") that induces a decomposition of $\mathbf{A}$ into $n$, rather than just $2$, factors. These algebras were called, naturally enough, \emph{Church algebras of dimension }$n$ ($n$CA), while algebras $\mathbf{A}$ all of whose elements induce an $n$-ary factorisation of this sort (i.e., are $n$-central) were given the name of \emph{Boolean-like algebras of
dimension }$n$ ($n\mathrm{BA}$s). Free $\mathcal V$-algebras (for $\mathcal V$ a variety), lambda algebras, semimodules over semirings -- hence, in particular, Boolean vector spaces -- give rise to Church algebras which, in general, have dimension greater than 2.

Varieties of $n\mathrm{BA}$s share many remarkable properties with the variety of Boolean algebras. We showed that any variety of $n\mathrm{BA}$s is generated by the $n\mathrm{BA}$s of cardinality $n$.  In the pure case (i.e., when the type includes the generalised if-then-else $q$, the $n$ constants, and nothing else), the variety is generated by a unique algebra $\mathbf{n}$ of universe $\{\e_1,\dots,\e_n\}$, so that any pure $n\mathrm{BA}$ is, up to isomorphism, a subalgebra  of $\mathbf{n}^I$, for a suitable set $I$. Another remarkable property of the $2$-element Boolean algebra is the definability of all finite Boolean functions in terms e.g. of the connectives {\sc and, or, not}. This property is inherited by the algebra $\mathbf{n}$:  all finite functions on $\{\e_1,\dots,\e_n\}$ are term-definable, so that the variety of pure $n\mathrm{BA}$s is primal. More generally, a variety of an arbitrary type with one generator is primal if and only if it is a variety of $n\mathrm{BA}$s.

In the present paper, we explore further the connections between Boolean algebras of finite dimension and primal algebras, as well as other consolidated topics in universal algebra, like the theory of \emph{Boolean powers}. Although Boolean powers, for which see e.g. \cite{Burris, Pinus}, emerged in the literature in the context of individual classes of algebras, like rings \cite{AK48}, Post algebras \cite{Rosenbloom} or Banach algebras \cite{Gelfand}, they were given a systematic and general treatment in the work of Foster \cite{Fa}. This construction is important and useful in that it permits to transfer properties of Boolean algebras to other, usually lesser-known, classes of algebras. Applications include answers to the decision problem for the first-order theories of certain varieties \cite{BM}, as well as representation results like Foster's celebrated theorem on primal varieties \cite[Thm. 7.4]{BS}, according to which any member of a variety generated by a primal algebra is a Boolean power of the generator.

In this paper, we parlay the theory of $n$-central elements into an extension to
arbitrary semirings of the technique of Boolean powers. We algebraically define the
semiring power $\mathbf{E}[R] $ of an arbitrary algebra $\mathbf{E}$ by a semiring $R$ as an algebra over the set of central elements of a certain semimodule. 
We obtain the following results:
\begin{itemize}
    \item If the semiring $R$ is a Boolean algebra, then the algebraically defined semiring power is isomorphic to the classical,  topologically defined, Boolean power \cite{BS}.
    \item For every semiring $R$, the semiring power $\mathbf{E}[R]$ is isomorphic to the Boolean power of $\mathbf E$ by the Boolean algebra of the complemented and commuting elements of $R$.
\end{itemize}
%Every semiring $R$ contains a Boolean algebra $B$ of the fully orthogonal and commuting elements such that the semiring power $\mathbf{A}[R]$ is isomorphic to the Boolean power of $\mathbf A$ by $B$.
We also prove two representation theorems. We first show that any pure $n$BA $\mathbf A$ contains a Boolean algebra $B_\mathbf A$; then we represent  $\mathbf A$, up to isomorphism, as the $n$BA of $n$-central elements of the Boolean vector space $(B_\mathbf A)^n$. 
%Due to this representation theorem, we can view the Boolean vector space $(B_\mathbf A)^n$ as a kind of linear approximation of $\mathbf{A}$. 
In the second representation theorem it is shown that any $n$BA in a variety of $n$BAs with one generator is isomorphic to a Boolean power, a result from which
Foster's theorem follows as a corollary.

\section{Preliminaries\label{benzina}}
The notation and terminology in this paper are pretty standard. For
concepts, notations and results not covered hereafter, the reader is
referred to \cite{BS,McKenzieMT87} for universal algebra and to \cite{Golan} for semirings
and semimodules. 
As to the rest, superscripts that mark the difference between operations and operation symbols will be dropped whenever the context is sufficient for a disambiguation. 
\subsection{Factor Congruences and Decomposition}
\begin{definition}
A tuple $(\theta_{1},\dots,\theta_n)$ of congruences on $\mathbf{A}
$ is a family of \emph{complementary factor congruences} if the function
$f:\mathbf{A}\rightarrow \prod\limits_{i=1}^n\mathbf{A}/\theta _{i}$
defined by $f(a) =(a/\theta _{1},\dots,a/\theta _{n})$ is an
isomorphism. When $\left\vert I\right\vert =2$, we say that $( \theta_{1},\theta _{2})$ is a pair of complementary factor congruences.
\end{definition}

A \emph{factor congruence} is any congruence which belongs to a family of
complementary factor congruences.

\begin{theorem}
\label{thm:cong} A tuple $(\theta_{1},\dots,\theta_n)$ of congruences on $\mathbf{A}$ is a family of complementary factor congruences exactly when:
\begin{enumerate}
\item $\bigcap_{1\leq i\leq n}\theta_{i}=\Delta$;

\item $\forall (a_1,\dots,a_n)\in A^n$, there is $u\in A$ such that $a_i\theta_{i}\,u$,
for all $1\leq i\leq n$.
\end{enumerate}
\end{theorem}

Therefore $(\theta_{1},\theta_{2})$ is a pair of complementary factor
congruences if and only if $\theta_{1}\cap\theta_{2}=\Delta$ and $\theta_{1}\circ\theta_{2}=\nabla$. The pair $(\Delta,\nabla)$ corresponds
to the product $\mathbf{A}\cong\mathbf{A}\times\mathbf{1}$, where $\mathbf{1}
$ is a trivial algebra; obviously $\mathbf{1}\cong\mathbf{A}/\nabla$ and 
$\mathbf{A}\cong\mathbf{A}/\Delta$. The set of factor congruences of $\mathbf A$
is not, in general, a sublattice of $\mathrm{Con}(\mathbf{A})$.

Factor congruences can be characterised in terms of certain algebra
homomorphisms called \emph{decomposition operators} (see \cite[Def.~4.32]{McKenzieMT87} for additional details). 
%For every $n$-sequence $\overline{x}$ and element $a\in A$ we denote by $\overline{x}[a/i]$ the $n$-sequence which coincides with~$\overline{x}$, except on $i$, where it takes the value $a$. 

\begin{definition}
\label{def:decomposition} A \emph{decomposition operator} on an algebra $\mathbf{A}$ is a function $f:A^{n}\rightarrow A$ satisfying the following
conditions: 
\begin{description}
\item[D1] $f( x,x,\dots,x) =x$;
\item[D2] $f( f( x_{11},x_{12},\dots,x_{1n}),\dots,f(x_{n1},x_{n2},\dots,x_{nn}))=f(x_{11},\dots,x_{nn})$;
\item[D3] $f$ is an algebra homomorphism from $\mathbf{A}^n$ to $ \mathbf{A}$.
\end{description}
\end{definition}

Axioms (D1)-(D3) can be equationally expressed.

There is a bijective correspondence between families of complementary factor
congruences and decomposition operators, and thus, between decomposition
operators and factorisations of an algebra.

\begin{theorem}
\label{prop:pairfactor} Any decomposition operator $f:\mathbf{A}^{n}\rightarrow \mathbf{A}$ on an algebra $\mathbf{A}$ induces a family of
complementary factor congruences $\theta _{1},\dots,\theta _{n}$, where each $\theta _{i}\subseteq A\times A$ is defined by: 
\begin{equation*}
a\ \theta _{i}\ b\ \ \text{iff}\ \ f(a,\dots,a,b,a,\dots,a)=a\qquad(\text{$b$ at position $i$}).
\end{equation*}
Conversely, any family $\theta _{1},\dots,\theta _{n}$ of complementary factor
congruences induces a decomposition operator $f$ on $\mathbf{A}$: 
$f(a_1,\dots,a_n)=u$ iff $a_{i}\,\theta _{i}\,u$, for all $i$, where such an element $u$ is provably unique.
\end{theorem}

\subsection{Semimodules and Boolean vector spaces}

A \emph{semiring} $R$ \cite{Golan} is an algebra $(R,+,\cdot ,0,1)$  such that $(R,+,0)$ is a commutative monoid, $(R,\cdot ,1)$ is a monoid, and the following equations hold: 
\begin{description}
\item[SR1] $x0= 0x= 0$;
\item[SR2] $x(y+z)=xy+yz$;
\item[SR3] $(y+z)x=yx+zx$.
\end{description}
Thus, in
particular, rings with unit are semirings in which every element has an additive
inverse. Thus, in particular, rings with unit are semirings in which every element has an additive
inverse. Actually, the set of all ideals of a given ring form an idempotent semiring under addition and multiplication of ideals. Also, any bounded distributive lattice $(L,\vee ,\wedge ,0,1)$ is a commutative, idempotent semiring under join and meet.

An $n$-tuple of elements $r_1,\dots,r_n$ of a semiring is \emph{fully orthogonal} if $r_1+\dots +r_n=1$ and $r_ir_j=0$ for every $i\neq j$.

\begin{definition}
If $R$ is a semiring, a \emph{(left) }$R$\emph{-semimodule} (see \cite{Golan})
is a commutative monoid $(V,+,\mathbf{0})$ for which we have a function $R\times V\rightarrow V$, denoted by $(r,\mathbf{v})\mapsto r\mathbf{v}$ and
called \emph{scalar multiplication}, which satisfies the following
conditions, for all elements $r,s\in R$ and all elements $\mathbf{v},\mathbf{w}\in V$:
\begin{description}
\item[SM1] $(r\cdot s)\mathbf{v}=r(s\mathbf{v});$
\item[SM2] $r(\mathbf{v}+\mathbf{w})=r\mathbf{v}+r\mathbf{w}$ and $(r+s)%
\mathbf{v}=r\mathbf{v}+s\mathbf{v};$
\item[SM3] $0\mathbf{v}=\mathbf{0}$ and $1\mathbf{v}=\mathbf{v}.$
\end{description}
\end{definition}

An $R$-semimodule is called: (i) a \emph{module} if $R$ is a ring;
(ii) a \emph{vector space} if $R$ is a field; (iii) a \emph{Boolean vector
space}  if $R$ is a Boolean algebra (see \cite{GL} for basic facts on Boolean vector spaces). The elements of an $R$-semimodule are called \emph{vectors}.

We present an example of an $R$-semimodule that fails, in general, to be a module.
\begin{example}\label{emmevusemiring}
Let $\mathbf{L} = ( L, \land, \lor, \cdot, \slash, \backslash, 0, 1 )$ be a bounded integral residuated lattice \cite{GJKO}. Then $R(\mathbf{L}) = (L, \lor, \cdot, 0, 1)$ is a semiring and $(L, \land, 1)$ is an $R(\mathbf{L})$-semimodule (see e.g. \cite{Gerla}). The scalar multiplication is given by $rv := v/r$, for all $r,v \in L$.
\end{example}

If $V$ is an $R$-semimodule and $E\subseteq V$ then we denote by

\[R\langle E\rangle =\{\sum_{i=1}^{n}r_{i}\e_{i}:r_{i}\in R,\e_{i}\in E,n\in \mathbb{N}\},\]

the set of linear combinations of elements of $E$. The set $E$ is called a 
\emph{free basis} of $V$ if $R\langle E\rangle =V$ and each vector in $V$
can be expressed as a linear combination of elements in $E$ in exactly one
way. An $R$-semimodule having a free basis $E$ is denoted by 
$R\langle E\rangle $ and is called the $R$\emph{-semimodule freely
generated by $E$}.

If $\mathbf{v}\in R\langle E\rangle $, then there exist $\e_{1},\dots ,\e_{n}\in E$ and scalars $r_{1},\dots ,r_{n}\in R$ such $\mathbf{v}=\sum_{i=1}^{n}r_{i}\e_{i}$. 

The \emph{coordinates} $v_{\mathsf{d}}$ ($\dd\in E$) of a vector $\mathbf{v}=\sum_{i=1}^{n}r_{i}\e_{i}$ w.r.t. the free basis $E$ are defined as follows:
$$v_{\mathsf{d}}=\begin{cases} 
r_i   &\text{if $\mathsf{d}= \e_{i}$ ($i=1,\dots,n$)}\\
 0 &\text{if $\mathsf{d}\neq \e_{1},\dots,\e_n$}.\\
\end{cases}
$$
Each element $\mathbf{v}$ of $R\langle E\rangle $ can be represented by the formal series $\mathbf{v}=\sum_{\e\in E}v_{\e}\e$, where almost all scalars $v_{\e}$ coincide with $0$.

The next example will play an important role in this paper.
\begin{example}
\label{exa:partition} (\emph{$n$-Sets}) Let $X$ be a set.
An \emph{$n$-subset} of $X$ is a sequence $(Y_{1},\dots ,Y_{n})$ of subsets 
$Y_{i}$ of $X$. We denote by $\mathrm{Set}_{n}(X)$ the family of all $n$-subsets of $X$.  $\mathrm{Set}_{n}(X)$ can be
viewed as the universe of a Boolean vector space over the powerset $\mathcal P(X)$
with respect to the following operations: 
$$(Y_{1},\dots ,Y_{n})+(Z_{1},\dots ,Z_{n})=(Y_{1}\cup Z_{1},\dots ,Y_{n}\cup
Z_{n})$$
and, for every $Z\subseteq X$, 
$$Z(Y_{1},\dots ,Y_{n})=(Z\cap Y_{1},\dots ,Z\cap Y_{n}).$$
$\mathrm{Set}_{n}(X)$ is freely generated by the $n$-sets $\e_{1}=(X,\emptyset ,\dots ,\emptyset ),\dots, \e_n=(\emptyset,\dots ,\emptyset ,X)$.  Thus, an arbitrary $n$-set $(Y_{1},\ldots ,Y_{n})$ has the canonical representation $Y_{1}\e_1+\dots +Y_{n}\e_n$ as a vector.
\end{example}

\subsection{Church algebras of finite dimension\label{dobbiaco}}

In \cite{SBLP} we introduced \emph{Church algebras of dimension }$n$, algebras having $n$ designated elements $\e_1,\dots,\e_n$ ($n\geq 2$) and an operation $q$ of arity $n+1$ (a sort of ``generalised if-then-else'') satisfying $q(\e_i,x_1,\dots,x_n)=x_i$.  The operator $q$ induces, through the so-called $n$-central elements, a decomposition of the algebra into $n$ factors.

\begin{definition}
\label{def:nCH}An algebra $\mathbf{A}$ of type $\tau $ is a \emph{Church algebra of
dimension} $n$ (an $n\mathrm{CH}$, for short) if there are term definable
elements $\e_{1}^{\mathbf{A}},\e_{2}^{\mathbf{A}},\dots ,\e_{n}^{\mathbf{A}}\in
A$ and a term operation $q^{\mathbf{A}}$ of arity $n+1$ such that, for all $%
b_{1},\dots ,b_{n}\in A$ and $1\leq i\leq n$, $q^{\mathbf{A}}(\e_{i}^{\mathbf{%
A}},b_{1},\dots ,b_{n})=b_{i}$. A variety $\mathcal{V}$ of type $\tau $ is a 
\emph{variety of algebras of dimension }$n$ if every member of $\mathcal{V}$
is an $n\mathrm{CH}$ with respect to the same terms $q,\e_{1},\dots ,\e_{n}$.
\end{definition}

%\begin{definition}
%\label{def:nCH} An algebra $\mathbf{A}$ of type $\tau $ is a \emph{dimensional algebra} (a $\mathrm{CA}$, for short) if 
%\begin{itemize}
%\item There are term definable
%elements $(\e_{i}^{\mathbf{A}}\in A : i\in X)$ and, for every $I=\{ \e_1^{\mathbf{A}},\dots,\e_n^{\mathbf{A}}\}\subseteq_{\mathrm{fin}} X$, a term operation $q_I^{\mathbf{A}}$ of arity $n+1$ such that, for all $b_{1},\dots ,b_{n}\in A$ and $1\leq i\leq n$, $q_I^{\mathbf{A}}(\e_{i}^{\mathbf{A}},b_{1},\dots ,b_{n})=b_{i}$. 
%\item If $I\subseteq J\subseteq_{\mathrm{fin}} X$, then 
%\end{itemize}
%
%
%A variety $\mathcal{V}$ of type $\tau $ is a 
%\emph{variety of algebras of dimension }$n$ if every member of $\mathcal{V}$
%is an $n\mathrm{CA}$ with respect to the same term $q$ and the same
%constants $\e_{1},\dots ,\e_{n}$.
%\end{definition}

{If $\mathbf{A}$ is an }$n\mathrm{CH}${,} then we call the algebra $\mathbf{A}_{0}=(A,q^{\mathbf{A}},\e^{\mathbf{A}}_{1},\dots ,\e_{n}^{\mathbf{A}})$ the \emph{pure reduct} of $\mathbf{A}$.

Church algebras of dimension $2$ were introduced as Church algebras in \cite{MS08} and studied in \cite{first}.  Examples of Church algebras of dimension $2$ are Boolean algebras (with $q(x,y,z) =(x\wedge z)\vee (\lnot x\wedge y)$) or rings with unit (with $q( x,y,z) =xz+y-xy$). Next, we list some relevant examples of Church algebras having dimension greater than $2${.}

%\begin{example}
%\label{libbera}(\emph{Free algebras)} Let $\tau $ be a type and $X=\{x_1,\dots,x_n\}$ be a finite set of variables. For all $\mathcal{\phi },\mathcal{\psi }_{1},\dots,\mathcal{\psi }_{n}$ in the term algebra $\mathbf{T}_{\tau }( X)$, we define:
%$$
%q( \mathcal{\phi },\mathcal{\psi }_{1},\dots,\mathcal{\psi }_{n}) =\mathcal{\phi }\left[ \mathcal{\psi }_{1}/x_{1},\dots,\mathcal{\psi }_{n}/x_{n}\right] \text{,}
%$$
%where $\mathcal{\phi }\left[ \mathcal{\psi }_{1}/x_{1},\dots,\mathcal{\psi }_{n}/x_{n}\right] $ denotes the term obtained from $\mathcal{\phi }$ by replacing each occurrence of $x_{i}$ by $\mathcal{\psi }_{i}$, for all $i$.
%If $\mathcal{V}$ is any variety of algebras of type $\tau $ and $\mathbf{T}_\mathcal{V}(X)$ is the free $\mathcal{V}$-algebra over $X$, this operation is well-defined on equivalence classes of terms in  $\mathbf{T}_\mathcal{V}(X)$ and turns  $\mathbf{T}_\mathcal{V}(X)$ into an $n\mathrm{CA}$, with respect to $q$ and $x_1,\ldots,x_n$.\end{example}

\begin{example}
\label{cocchio} (\emph{Semimodules})
Let $V$ be an $R$-semimodule freely generated by a finite set $E=\{\e_{1},\dots ,\e_n\}$. Then we define an operation $q$ of arity $n+1$ as
follows (for all $\mathbf{v}=\sum_{j=1}^nv_{j}\e_j$
and $\mathbf{w}^{i}=\sum_{j=1}^nw_{j}^{i}\e_j$):
$$
q(\mathbf{v},\mathbf{w}^{1},\dots ,\mathbf{w}^{n})=\sum_{i=1}^{n}v_{i}\mathbf{w}^{i}=\sum_{k=1}^n(\sum_{i=1}^{n}v_{i}\cdot w_{k}^{i})\e_k.
$$
Under this definition, $V$ becomes an $n\mathrm{CH}$.
Each $\mathbf{w}^{i}$ in the definition of $q$
can be viewed as the $i$-th column vector of an $n\times n$ matrix $M$. The
operation $q$ does nothing but express in an algebraic guise the
application of the linear transformation encoded
by $M$ to the vector $\mathbf{v}$.
\end{example}

A special instance of the previous example, described below, will play a role in what follows.

\begin{example}
\label{exa:partition2} (\emph{$n$-Sets}) Consider the Boolean vector space $\mathrm{Set}_{n}(X)$ from Example \ref{exa:partition}. Retaining the notation from that example, the $q$ operator defined for generic semimodules in Example \ref{cocchio} can be given here an explicit description as follows, for all $\mathbf{y}^{i}= Y^i_{1}\e_1+\dots +Y^i_{n}\e_n$:
$$
q( \mathbf{y}^0,\mathbf{y}^1,\dots ,\mathbf{y}^n) =(\bigcup\limits_{i=1}^{n}Y^0_{i}\cap Y_{1}^{i},\dots,\bigcup\limits_{i=1}^{n}Y^0_{i}\cap Y_{n}^{i}).
$$
\end{example}

{In \cite{vaggione}, D. Vaggione introduced the notion
of \emph{central element} in order to study algebras whose complementary factor
congruences can be replaced by certain elements of their universes. Usually, if a
manageable characterization of such elements is available, one gets important
insights into the structure theories of the algebras at issue. To list a few
examples, central elements coincide with central idempotents in rings with
unit, with complemented elements in $FL_{ew}$-algebras, which form the equivalent algebraic semantics of the full Lambek calculus with exchange and weakening, and with members of
the centre in ortholattices. In \cite{first}, T. Kowalski and three of the
present authors investigated central elements in Church algebras of dimension }$2$. 
In \cite{SBLP}, the idea was generalised to Church algebras of arbitrary finite dimension. 

%If $a,b$ are two elements of an algebra $\mathbf{A}$, we denote by $\theta(a,b)$ the least congruence $\psi$ such that $(a,b)\in\psi$.

\begin{definition}
\label{def:ncentral} If $\mathbf{A}$ is an $n\mathrm{CH}$, then $c\in A$
is called \emph{$n$-central} if the sequence of congruences $(\theta (c,\e_{1}),\dots
,\theta (c,\e_{n}))$ is an $n$-tuple of complementary factor congruences of $\mathbf{A}$. A central element $c$ is \emph{nontrivial} if $c\notin\{\e_{1},\dots ,\e_{n}\}$.
\end{definition}

The following characterisation of $n$-central elements, as well as the subsequent elementary result about them, were also proved in \cite{SBLP}.

\begin{theorem}
\label{thm:centrale} If $\mathbf{A}$ is an $n\mathrm{CH}$ of type $\tau $
and $c\in A$, then the following conditions are equivalent:

\begin{enumerate}
\item $c$ is $n$-central;

\item $\bigcap_{i\leq n}\theta (c,\e_{i})=\Delta $;

\item for all $a_{1},\dots ,a_{n}\in A$, $q(c,a_{1},\dots ,a_{n})$ is the
unique element such that $a_{i}\ \theta (c,\e_{i})\ q(c,a_{1},\dots ,a_{n})$,
for all $1\leq i\leq n$;

\item The following conditions are satisfied:
\begin{description}
\item[B1] $q(c,\e_{1},\dots ,\e_{n})=c$.

\item[B2] $q(c,x,x,\dots ,x)=x$ for every $x\in A$.

\item[B3] If $\sigma \in \tau $ has arity $k$ and $\mathbf x$ is a $n\times k$ matrix  of elements of $A$, then\\
$
q(c,\sigma (\mathbf x_{1}),\dots ,\sigma (\mathbf x_{n}))=\sigma (q(c,\mathbf x^1),\dots ,q(c,\mathbf x^k)).$
\end{description}
\item The function $f_{c}$, defined by $f_{c}(x_{1},\dots
,x_{n})=q(c,x_{1},\dots ,x_{n})$, is a decomposition operator on $\mathbf{A}$ such that 
$f_{c}(\e_{1},\dots ,\e_{n})=c.$
\end{enumerate}
\end{theorem}

For any  $n$-central element $c$ and any $n\times n$ matrix $\mathbf x$  of elements of $A$, a direct consequence of (B1)-(B3) gives
\begin{description}
\item[B4]
$q(c,q(c,\mathbf x_{1}),\dots, q(c,\mathbf x_{n})) =q(c, x^1_1,x^2_{2}, \dots, x^n_{n})$.
\end{description}
For any $n\times (n+1)$ matrix $\mathbf y$  of elements of $A$, another consequence of (B1)-(B3) gives 
$$ q(c,q(\mathbf y_{1}),\dots, q(\mathbf y_{n}))=q(q(c,\mathbf y^0),q(c,\mathbf y^1),\dots ,q(c,\mathbf y^n)).$$

%$q(c,q(c,x_1,x_2), q(c,y_1,y_2)) =q( )$.

\begin{proposition}\label{prop-closure} Let $\mathbf{A}$ be an $n\mathrm{CH}$. Then the set of all $n$-central elements of $\mathbf{A}$ is a subalgebra of the pure reduct of $\mathbf{A}$.
\end{proposition}

Hereafter, we denote by $\mathbf{Ce}_{n}(\mathbf{A})$ the algebra 
$(\mathrm{Ce}_{n}(\mathbf{A}),q,\e_{1},\dots ,\e_{n})$ 
of all $n$-central elements of an 
$n\mathrm{CH}$ $\mathbf{A}$.

\bigskip

\begin{example}
Let $\mathbf A$ be an arbitrary algebra (not necessarily an $n\mathrm{CH}$) of  type $\tau$ and $F$ be the set of all functions from $A^n$ into $A$. Consider the $n\mathrm{CH}$  
$\mathbf F = (F,\sigma^\mathbf F,q^\mathbf F,\e_1^\mathbf F,\dots,\e_n^\mathbf F)_{\sigma\in\tau}$,
whose operations are defined as follows (for all $f_i, g_j\in F$ and all $a_1,\dots,a_n\in A$):
\begin{enumerate}
\item $\e_i^\mathbf F(a_1,\dots,a_n)=a_i$;
\item $q^\mathbf F(f,g_1\dots,g_n)(a_1,\dots,a_n) =f(g_1(a_1,\dots,a_n)\dots,g_n(a_1,\dots,a_n))$;
\item $\sigma^\mathbf F(f_1,\dots,f_k)(a_1,\dots,a_n) = \sigma^\mathbf A(f_1(a_1,\dots,a_n),\dots, f_k(a_1,\dots,a_n))$, for every $\sigma\in\tau$ of arity $k$.
\end{enumerate}
%For any $a\in A$, let $f_a:A^n\to A$ be the constant function of value $a$.
%The map $a\in A\mapsto f_a\in F$, where $f_a(x_1,\dots,x_n)=a$ for all $x_i\in A$, defines an embedding of $\mathbf A$ into the $\tau$-reduct of $\mathbf F$.
Let $\mathbf G$ be any subalgebra of $\mathbf F$ containing all constant functions.
It is possible to prove that a function $f:A^n\to A$ is an $n$-central element of $\mathbf G$ if and only if it is an $n$-ary decomposition operator of the algebra $\mathbf A$ commuting with every element $g\in G$ (for every $a_{ij}\in A$):
$$f(g(a_{11},\dots, a_{1n}),\dots, g(a_{n1},\dots, a_{nn})) = g(f(a_{11},\dots, a_{n1}),\dots, f(a_{1n},\dots, a_{nn})).$$
The reader may consult \cite{SLP18} for the case $n=2$.
\end{example}

The following example provides an application of the $n$-central elements  to lambda calculus.

\begin{example}
\label{exa:lambda} (\emph{Lambda Calculus}) We refer to Barendregt's book \cite{Bare} for basic definitions on lambda calculus. Let $\mathrm{Var}=\{a,b,c,a_1,\dots\}$ be
the set of variables of the lambda calculus ($\lambda $-variables, for short), and let $x,y,x_{1},y_{1},\dots$ be the algebraic variables (``holes'' in the terminology of Barendregt's book). Once a finite set $I=\{a_{1},\dots ,a_{n}\}$ of $\lambda $-variables has been fixed,  we define $n$ constants $\e_i$ and an operator $q_{I}$ as follows: 
$$\e_{i}= \lambda a_1\dots a_n.a_i;\qquad q_{I}(x,y_{1},\dots ,y_{n})=(\dots ((xy_{1})y_{2})\dots )y_{n};$$
The term algebra of a $\lambda$-theory 
 is an $n\mathrm{CH}$ with respect to the term operation $q_{I}$ and the constants $\e_{1},\dots ,\e_{n}$. 
  Let $\Lambda_\beta$ be the term algebra of the minimal $\lambda$-theory $\lambda\beta$, whose  lattice of  congruences is isomorphic to the lattice of $\lambda$-theories.  Let $\Omega = (\lambda a.aa)(\lambda a.aa)$ be the canonical looping $\lambda$-term. It turns out that $\Omega$   can be consistently equated to any other closed $\lambda$-term (see \cite[Proposition 15.3.9]{Bare}).
Consider the $\lambda$-theory $T_i = \theta(\Omega, \e_i)$ generated by equating $\Omega$ to the constant $\e_i$ above defined. By Theorem \ref{thm:centrale} $\Omega$ is a nontrivial central element in the term algebra of the $\lambda$-theory $T=\bigcap_{i\leq n}\theta (\Omega, \e_i)$, i.e., the quotient $\Lambda_\beta/ T$. It is possible to prove that any model of $T$, not just its term model, is decomposable. On the other hand, the set-theoretical models of $\lambda$-calculus defined after Scott's seminal work \cite{Scott} are indecomposable algebras. Hence, none of these set-theoretical models has $T$ as equational theory. This general incompleteness result has been proved in \cite{MS06} for the case $n=2$ (see also \cite{MS10}).
\end{example}

\subsection{Boolean-like algebras of finite dimension\label{trallallero}}

Boolean algebras are a prototypical example of Church algebras of finite dimension, namely of dimension $2$, in which all elements are $2$-central. It turns out that, among the $n$-dimensional Church algebras, those
algebras all of whose elements are $n$-central inherit many of the
remarkable properties that distinguish Boolean algebras. 

\begin{definition}
\label{mezzucci}An $n\mathrm{CH}$ $\mathbf{A}$ is called a \emph{Boolean-like
algebra of dimension $n$} ($n\mathrm{BA}$, for short) if every element of $A$
is $n$-central.
\end{definition}

By Proposition \ref{prop-closure}, the algebra $\mathbf{Ce}_{n}(\mathbf{A})$ of all $n$-central
elements of an $n\mathrm{CH}$ $\mathbf{A}$ is a canonical example of  $n\mathrm{BA}$. The class of all $n\mathrm{BA}$s of type $\tau $ is a variety of $n$CHs axiomatised by the identities
B1-B3 in Theorem \ref{thm:centrale}.

\begin{example}\label{exa:n}
The algebra 
$\mathbf{n}=( \{ \mathsf \e_{1},\dots,\mathsf \e_{n}\} ,q^{\mathbf{n}},\mathsf e^\mathbf{n}_{1},\dots,\mathsf e^\mathbf{n}_{n})$,
where $$q^{\mathbf{n}}( \mathsf \e_{i},x_{1},\dots,x_{n}) =x_{i}$$ for every $i\leq n$, is a pure $n\mathrm{BA}$.
\end{example}

\begin{example}\label{exa:parapa} (\emph{$n$-Partitions}) Let $X$ be a set. An \emph{$n$-partition} of $X$ is an $n$-subset $(Y_{1},\ldots ,Y_{n})$ of $X$ such that $\bigcup_{i=1}^{n}Y_{i}=X$ and $Y_{i}\cap Y_{j}=\emptyset $ for all $i\neq j$.
The set of $n$-partitions of $X$ is closed under the $q$-operator defined in 
Example \ref{exa:partition} and constitutes  the algebra of all $n$-central
elements of the Boolean vector space $\mathrm{Set}_{n}(X)$ of all $n$-subsets of $X$. 
Notice that the algebra  of $n$-partitions of $X$, denoted by $\mathrm{Par}(X)$, is isomorphic to the $n\mathrm{BA}$ $\mathbf{n}^X$.
\end{example}

\bigskip

Several remarkable properties of Boolean algebras find some analogue in the structure theory of $n\mathrm{BA}$s. 

\begin{theorem}
\label{lem:subirr} 
\begin{enumerate}
    \item An $n\mathrm{BA}$ $\mathbf{A}$ is subdirectly irreducible if and only if $|A|=n$.
    \item Any variety $\mathcal{V}$ of $n\mathrm{BA}$s is generated by the finite set $\{\mathbf{A}\in \mathcal{V}:|A|=n\}$. In particular, the variety of pure $n\mathrm{BA}$s is generated by the algebra $\mathbf{n}$.
    \item If an $n\mathrm{BA}$ $\mathbf{A}$  has a minimal subalgebra $\mathbf{E}$ of cardinality $n$, then ${V}(\mathbf{A})={V}(\mathbf{E})$.
    \item Every $n\mathrm{BA}$ $\mathbf{A}$ is isomorphic to a subdirect product of $\mathbf{B}_{1}^{I_{1}}\times \dots \times \mathbf{B}_{k}^{I_{k}}$ for some sets $I_{1},\dots ,I_{k}$ and some $n\mathrm{BA}$s $\mathbf{B}_{1},\dots ,\mathbf{B}_{k}$ of cardinality $n$.
    \item Every pure $n\mathrm{BA}$ $\mathbf{A}$ is isomorphic to a subdirect power of $\mathbf{n}^{I}$, for some set $I$.
\end{enumerate}
\end{theorem}

%Observe that any variety $\mathcal{V}$ of $n\mathrm{BA}$s is a
%discriminator variety, whose switching term $s$ is
%\begin{equation*}
%s(x,y,z,w)=q( x,q(y,\overline{w}[ z/1] ),\dots,q(y,\overline{w}[ z/n] )) \text{,}
%\end{equation*}
%where  the $n$-sequence $\overline{w}[ z/i]$ coincides with~$\overline{w}$, except on $i$,
%where it takes the value $z$.

However, we cannot assume that any two $n$-element algebras in an arbitrary variety
$\mathcal{V}$ of $n$BAs are isomorphic, for such algebras may have further operations
over which we do not have any control.

%\begin{lemma}\label{lem:alg-npart} If $I$ is a set, then
%the algebra $\mathbf{\wp_{n}}(I)=(\wp_{n}(X),q^{\wp},e_1^{\wp},\dots, e_n^{\wp})$ is isomorphic to the  $\mathrm{nBA}$ $\mathbf{n}^{I}$.
%\end{lemma}
%
%\begin{proof} We define a map $h$ from $\hat n^{X}$ to $\wp_{n}(X)$. To each function $a: X\to \hat n$ we associate the $n$-partition $h(a)=A$ such that $A_j = a^{-1}(\{j\})$.
%
%Set $h(a)=A,h(b^{1})=B^{1},\ldots, h(b^{n})=B^{n}$ and $P=h(q^{\mathbf{n}^{X}}(a,b^{1}\ldots,b^{n}))$.
%We confine ourselves in proving that $P=B\cdot A$. Let $x\in X$. It can be seen that $q^{\mathbf{n}^{X}}(a,b^{0}\ldots,b^{n-1})_{(x)}=q^{\mathbf{n}}({a}_{x},b^0_{x}\ldots,{b^{n-1}_{x}})=b^{{a}_{x}}_{x}\in n$. Therefore, by definition of $h$, $x$ belong to the partition $P$ in the coordinate whose index is $b^{{a}_{x}}_{x}$. Let us observe that $x$ belongs to both $A_{a_{x}}$ and $B_{b^{{a}_{x}}_{x}}^{a_{x}}$, and so to $\bigcup_{k=0}^{n-1} A_{k}\cap B_{b^{{a}_{x}}_{x}}^{k}$, which equals to the ${b^{{a}_{x}}_{x}}$-coordinate of the partition $B\cdot A$. The fact that $h$ preserves the constants is straightforward.
%\end{proof}

A subalgebra of the $n$BA $\mathrm{Par}(X)$ of the $n$-partitions on a set $X$, defined in Example \ref{exa:parapa},  is called a \emph{field of $n$-partitions on $X$}. The Stone representation theorem for $n$BAs follows.

\begin{corollary} \label{cor:field-partitions}
Any pure $\mathrm{nBA}$ is isomorphic to a field of $n$-partitions on a suitable set $X$.
\end{corollary}

\section{Applications of $n$BAs to Boolean powers}
For an algebra $\mathbf E$ of a given type and a Boolean algebra $B$, the Boolean power of $\mathbf E$ by $B$ is the algebra $\mathcal C(B^*,\mathbf E)$ of all continuous functions from the Stone space $B^*$ of $B$ to $E$, where $E$ is given the discrete topology and the operations of $E$ are lifted to $\mathcal C(B^*,\mathbf E)$ pointwise (see, e.g., \cite{BS}). Boolean powers can be traced back to the work of Foster \cite{Fa} in 1953 and turned out to be a very useful tool in universal algebra
 for exporting properties of Boolean algebras into other varieties. 

The continuous functions from $B^*$ to a discrete space $E$ determine finite partitions of $B^*$ by clopen sets. By Corollary \ref{cor:field-partitions} such partitions can be canonically endowed with a structure of $n$BA. Thus it is natural to algebraically rephrase   the theory of Boolean powers by using central elements and, by the way, to generalise Boolean powers to arbitrary semiring powers. 
We define the semiring power $\mathbf{E}[R] $ of an arbitrary algebra $\mathbf{E}$ by a semiring $R$ as an algebra whose universe is the set of  central elements of a certain semimodule; the operations of $\mathbf E$ are extended by linearity.

We obtain the following results:
\begin{enumerate}
\item[(i)] The central elements of an $R$-semimodule are characterised by those finite sets of elements of $R$ that are fully orthogonal, idempotent and commuting (Theorem \ref{prop:semimodulecentral}).
\item[(ii)] If $R$ is a Boolean algebra, then the algebraically defined semiring power $\mathbf{E}[R]$ is isomorphic to the Boolean power $\mathcal C(R^*, \mathbf E)$ (Theorem \ref{prop:Ebp}).
\item[(iii)] For every semiring $R$, the semiring power $\mathbf{E}[R]$ is isomorphic to the Boolean power $\mathcal C(C(R)^*,\mathbf E)$ of $\mathbf E$ by the Boolean algebra $C(R)$ of complemented and commuting elements of $R$  (Theorem \ref{thm:semiboolean}).
\end{enumerate}

%More precisely, given a $\nu $-algebra $\mathbf{E}$ and a semiring $R$, we describe the semiring
%power of $\mathbf{E}$ by $R$ in three steps:
%\begin{enumerate}
%\item We expand the algebra $\mathbf{E}$ to the $R$-semimodule $
%V$, freely generated by $E$, in such a way that every
%operation $g^V$ is multilinear and
%coincides with $g^{\mathbf{E}}$ over the generators.
% 
%\item We consider the algebra $V_{\mathbf{E}} = (V, g^{V}, q^V_I,e)_{g\in \nu, I\subseteq_{\mathrm{fin}}E,e\in E}$, where $q_I$ is defined in Example \ref{cocchio}. We characterise
% the elements of $V$ which are $I$-central for some such $q_{I}$, and collect them into a set $E[R]$.
%
%\item We show that
%\begin{itemize}
%\item[(i)] $R[E]$ is a subalgebra of $(V ,g^V)_{g\in \nu }$, which is an $n$BA whenever $E$ has finite cardinality $n$.
%\item[(ii)] If $R$ is a Boolean algebra, then $R[E]$ coincides with the Boolean power of $\mathbf E$ by $R$.
%\end{itemize}
%\end{enumerate}
%Let us now get down to the nitty-gritty. 

Let $\mathbf{E}$ be an algebra of
type $\tau $ and $R$ be a (non necessarily commutative) semiring. Consider the $R$-semimodule $V$ freely generated by the set $E$.  A vector of $V$ is a linear combination $\mathbf{v}=\sum_{\e\in
E}v_{\e}\e$, where $v_{\e}\in R$ is a scalar and $v_{\e}=0$ for all
but finitely many $\e\in E$. Let us observe that $\sum_{\e\in
E}v_{\e}\e$\emph{ always converges}, since it involves only finitely many $e\in E$.\\

Every operation 
$g^\mathbf E$ 
of the algebra $\mathbf E$ can be linearly lifted to an operation $g^V$ on $V$:
\[
g^{V}(\mathbf{v}^{1},\dots ,\mathbf{v}^{k})=\sum_{\dd_{1},\dots,\dd_{k}\in E}(v_{\dd_{1}}^{1}\cdot \ldots \cdot
v_{\dd_{k}}^{k})g^{\mathbf{E}}(\dd_{1},\dots ,\dd_{k})
\]
 For every finite nonempty $I=\{\e_{1},\dots ,\e_{n}\}\subseteq E$, we define (cf. Example \ref{cocchio}):
 $$q^V_I(\mathbf{v},\mathbf{w}^{1},\dots,\mathbf{w}^{n}) = \sum_{i=1}^n
v_{\e_i}\mathbf{w}^{i} = \sum_{\dd\in E} (\sum_{i=1}^n v_{\e_i}  w^i_\dd)\dd, $$
and we call a vector $\mathbf{v}\in V$ \emph{$I$-central} if it is an $n$-central element with respect to
the operation $q_{I}$ in the algebra 
$$V_{\mathbf{E}} = (V, +^{V}, f^{V}_r,g^{V}, q^V_J,\e^V)_{r\in R,g\in \tau,J\subseteq_{\mathrm{fin}}E,\e\in E},$$ 
where $f^{V}_r$ is the unary operation defined by $f^V_r(\mathbf{v})=r\mathbf{v}$.\\

The next Theorem shows that a vector is  $I$-central in $V_{\mathbf{E}}$ if and only if it can be written uniquely as an $R$-linear combination of fully orthogonal and commuting idempotents. 

\begin{theorem}
\label{prop:semimodulecentral} Let $I=\{\e_{1},\dots ,\e_{n}\}\subseteq E$. A vector $\mathbf{a}=\sum_{\dd\in
E}a_{\dd}\dd \in V$ is $I$-central in the
algebra $V_{\mathbf{E}}$ iff the following conditions are satisfied:

\begin{itemize}
\item[(i)] $a_\dd= 0$ for all $\dd\in E\setminus I$;

\item[(ii)] $a_{\e_1}+\dots+a_{\e_n}=1$;

\item[(iii)] $a_{\e_i} x = x a_{\e_i}$, for all $x\in R$ and $\e_i\in I$;

\item[(iv)] $a_{\e_{i}} a_{\e_{j}}=
\begin{cases}
0 & \text{if }i\neq j \\ 
a_{\e_{i}} & \text{if }i=j
\end{cases}
\quad$ for all $\e_{i},\e_{j}\in I$.
\end{itemize}
\end{theorem}

\begin{proof} 
  ($\Leftarrow$) We check identities (B1)-(B3) of Theorem \ref{thm:centrale}  for an element $\mathbf a\in V$ satisfying  hypotheses (i)-(iv). Since no danger of confusion will be impending, throughout this proof we write $a_i$ for $a_{\e_i}$ ($\e_i\in I$). 
%Notice that from (i) it follows that condition \eqref{eq:qI} of page \pageref{eq:qI} simplifies to $q(\mathbf a, \mathbf w^1,\dots, \mathbf w^n) = a_1 \mathbf w^1 +\dots+ a_n \mathbf w^n$ for all $\mathbf w^i\in V$.

 B1: $q_I(\mathbf a, \e_1,\ldots, \e_n)=a_{1} \e_1+\dots+a_{n} \e_n = \mathbf a$, by (i).

 B2: $q_I(\mathbf a,\mathbf b,\ldots,\mathbf b)=a_1\mathbf b+\dots+a_n\mathbf b =(\sum_{j=1}^{n} a_j) \mathbf b= 1\mathbf b$, by (ii), and $1\mathbf b=\mathbf b$.

 B3: 

\begin{align*}
q_I(\mathbf a,\mathbf w^1+\mathbf v^1,\ldots,\mathbf w^n+\mathbf v^n)&= \sum_{i=1}^{n} a_i(\mathbf w^i+\mathbf v^i)\\
&=\sum_{i=1}^{n} a_i\mathbf w^i + \sum_{i=1}^{n} a_i\mathbf v^i\\
&=q_I(\mathbf a,\mathbf w^1,\ldots,\mathbf w^n)+q_I(\mathbf a,\mathbf v^1,\ldots,\mathbf v^n)
\end{align*}

%u=\sum_{l=1}^{n} a_i(\mathbf w^i+\mathbf v^i)
%= \sum_{l=1}^{n} a_i\mathbf w^i + \sum_{l=1}^{n} a_i\mathbf v^i
%= q_I(\mathbf a,\mathbf w^1,\ldots,\mathbf w^n)+q_I(\mathbf a,\mathbf v^1,\ldots,\mathbf v^n)
%$$
%and
$$
q_I(\mathbf a,r\mathbf w^1,\ldots,r\mathbf w^n)    =   \sum_{i=1}^{n} a_i(r\mathbf w^i) 
    =  \sum_{i=1}^{n} (a_i r)\mathbf w^i 
    =\sum_{i=1}^{n} (r a_i)\mathbf w^i$$
by (iii), and 
$\sum_{i=1}^{n} (r a_i)\mathbf w^i=rq_I(\mathbf a,\mathbf w^1,\ldots,\mathbf w^n)$.

Without loss of generality, we assume $g\in \tau$ to be a binary operator:
% and we put $w = g^V(q_I(\mathbf v,\mathbf t^1,\dots,\mathbf t^n),q_I(\mathbf v,\mathbf w^1,\dots,\mathbf w^n))$:

\begin{align*}
g^V(q_I(\mathbf a,\mathbf t^1,\dots,\mathbf t^n),q_I(\mathbf a,\mathbf w^1,\dots,\mathbf w^n)) &= g^V(a_1\mathbf t^1+\dots+a_n\mathbf t^n,a_1\mathbf w^1+\dots+a_n\mathbf w^n )\\
&=\sum_{i=1}^n \sum_{j=1}^n (a_i a_j) g^V(\mathbf t^i,\mathbf w^j)\text{ (linearity)}\\
&=\sum_{i=1}^n a_i g^V(\mathbf t^i,\mathbf w^i)\text{ (iv)}\\
&=q_I(\mathbf a,g^V(\mathbf t^1,\mathbf w^1),\dots, g^V(\mathbf t^n,\mathbf w^n)).\\
\end{align*}

 Let $\mathbf c^1,\dots,\mathbf c^n\in V$, $J=\{\dd_1,\dots,\dd_k\}\subseteq E$, and $ Y$ be a $k\times n$ matrix of vectors of $V$. We denote by $Y^j$ the $j$-th column of $Y$ and by $Y_i$ its $i$-th row.
If $1\leq s\leq n$ then we have:
 $q_J(\mathbf c^s, Y^s) =q_J(\mathbf c^s, Y^s_1,\dots, Y^s_k) = \sum_{l=1}^k    c^s_{\dd_l} Y^s_l$
 and 
 \begin{equation}\label{nino1}q_I(\mathbf a,q_J(\mathbf c^1, Y^1),\ldots,q_J(\mathbf c^{n}, Y^n))=
 \sum_{s=1}^{n} a_s q_J(\mathbf c^s, Y^s)= \sum_{l=1}^{k}\sum_{s=1}^{n} (a_s  c^s_{\dd_l})  Y^s_l.\end{equation}
% \[
%\begin{array}{lll}
%q_I(\mathbf a,q_J(\mathbf c^1,\mathbf Y^1),\ldots,q_J(\mathbf c^{n},\mathbf Y^n))  &  = &  \sum_{s=1}^{n} a_s q(\mathbf c^s,\mathbf Y^s) \\
%  & =  & \sum_{s=1}^{n}a_s ( \sum_{l=1}^k    c^s_{h_l}\mathbf Y^s_l)  \\
%  &  = &   \sum_{l=1}^{k}\sum_{s=1}^{n} (a_s\cdot  c^s_{h_l}) \mathbf Y^s_l% +\sum_{d\in E\setminus I}(\sum_{s=1}^{n}a_s\cdot c^s_d)  \dd\\
%\end{array}
%\]
%We now show that $q_I(\mathbf a,q_J(\mathbf c^1,\mathbf Y^1),\ldots,q_J(\mathbf c^{n},\mathbf Y^n))   =       q_J(q_I(\mathbf a,\mathbf c^1,\ldots,\mathbf c^{n}),q_I(\mathbf a,\mathbf Y_1),\ldots,q_I(\mathbf a,\mathbf Y_k))$.
%\[
%\begin{array}{rrr}
%q_I(\mathbf a,q_J(\mathbf c^1,\mathbf Y^1),\ldots,q_J(\mathbf c^{n},\mathbf Y^n))  &  = &   \\
%    q_J(q_I(\mathbf a,\mathbf c^1,\ldots,\mathbf c^{n}),q_I(\mathbf a,\mathbf Y_1),\ldots,q_I(\mathbf a,\mathbf Y_k)) &&   \\
%\end{array}
%\]

%Let $t=q_J(q_I(\mathbf a,\mathbf c^1,\ldots,\mathbf c^{n}),q_I(\mathbf a,\mathbf Y_1),\ldots,q_I(\mathbf a,\mathbf Y_k))$. 
As
\begin{align*}
q_I(\mathbf a,\mathbf c^1,\ldots,\mathbf c^{n})&=\sum_{j=1}^{n} a_j \mathbf c^j\\
&= \sum_{j=1}^{n} a_j (\sum_{\ch\in E}c^j_\ch \ch)\\
&=\sum_{\ch\in E} (\sum_{j=1}^{n} a_jc^j_\ch)\ch
\end{align*}

and $J=\{\dd_1,\dots,\dd_k\}$, we get the conclusion by applying (iii)-(iv) and (\ref{nino1}):

\begin{align*}
q_J(q_I(\mathbf a,\mathbf c^1,\ldots,\mathbf c^{n}),q_I(\mathbf a, Y_1),\ldots,q_I(\mathbf a, Y_k)) &=\sum_{l=1}^{k} q_I(\mathbf a,\mathbf c^1,\ldots,\mathbf c^{n})_{\dd_l}  q_I(\mathbf a, Y_l)\\ 
 %+ \sum_{d\in E\setminus I}q(\mathbf a,\mathbf c^1,\ldots,\mathbf c^{n})_d  \dd
&=\sum_{l=1}^{k}( [\sum_{j=1}^{n} a_j c^j_{\dd_l} ][\sum_{s=1}^{n} a_s  Y_l^s])\\ 
  %+ \sum_{d\in E\setminus I}(\sum_{s=1}^{n}a_s\cdot c^s_d)  \dd
&=\sum_{l=1}^{k}(\sum_{s=1}^{n}\sum_{j=1}^{n} a_j c^j_{\dd_l}  a_s)  Y_l^s\\
&=\sum_{l=1}^{k}\sum_{s=1}^{n}(a_s c^s_{\dd_l})  Y_l^s\\%\text{by (iii)-(iv)}
&=q_I(\mathbf a,q_J(\mathbf c^1, Y^1),\ldots,q_J(\mathbf c^{n}, Y^n)).
\end{align*}

($\Rightarrow$) Since $\mathbf a$ is $I$-central, then by (B1) we have that $\mathbf a=q(\mathbf a, \e_1,\ldots, \e_n)=\sum_{i=1}^n a_i \e_i$. This implies (i). Item (ii) follows by (B2), because
$1\e_1=\e_1=q_I(\mathbf a,\e_1,\ldots,\e_1)=(\sum_{i=1}^n a_i)\e_1$ and each vector in $V$
can be expressed as a linear combination of elements in $E$ in exactly one way.
%If $\mathbf b =  \e_1$ then $(\sum_{i=1}^{n} a_i) \e_1=  \e_1$. 
%This implies (ii): $\sum_{i=1}^{n} a_i = 1^R$.
%and $a_d=0^R$ for all $d\in E\setminus I$, because every element of $V$ can be univocally expressed as linear combination of the generators. 

To show (iii) we consider the following chain of equalities:

\begin{align*}
(r a_1) \e_1 &=r(a_1 \e_1)\\
&=rq_I(\mathbf a, \e_1,\mathbf 0,\ldots,\mathbf 0)\\
&= q_I(\mathbf a,r \e_1,r\mathbf 0\ldots,r\mathbf 0)\text{ (B3)}\\
&= a_1(r \e_1)\\
&= (a_1r) \e_1
\end{align*}
It follows that $ra_1 = a_1r$
for every $r\in R$. Similarly for the other coordinates of $\mathbf a$.

From (B3)
\[
q_I(\mathbf a,q_I(\mathbf c^1, Y^1),\ldots,q_I(\mathbf c^{n}, Y^n))   = 
  q_I(q_I(\mathbf a,\mathbf c^1,\ldots,\mathbf c^{n}),q_I(\mathbf a, Y_1),\ldots,q_I(\mathbf a, Y_n)).
\]
It follows that 
$$\sum_{s=1}^{n} a_s\left(\sum_{l=1}^{n} c^s_l  Y^s_l\right)=
\sum_{l=1}^{n}\sum_{s=1}^{n} (a_s  c^s_l)  Y^s_l
%+\sum_{d\in E\setminus I}(\sum_{s=1}^{n}a_s\cdot c^s_d)  \dd
=_{\mathrm{B3}} \sum_{l=1}^{n}\sum_{s=1}^{n}(\sum_{j=1}^{n} a_j c^j_l  a_s)  Y_l^s
%+ \sum_{d\in E\setminus I}(\sum_{s=1}^{n}a_s\cdot c^s_d)  \dd
.$$
Fix $\e\in I$, a row $l$ and column $s$. Let  $ Y^s_l = \mathbf e$ and $ Y^i_r = \mathbf 0$ for all $(i,r)\neq (s,l)$. Then
$$a_s  c^s_l = \sum_{j=1}^{n} a_j c^j_l  a_s.$$
We get $a_s =a_s a_s$ of item (iv) by putting $c^s_l = 1$ and $c^j_l =  0$ for all $j\neq s$.
The last condition $a_j  a_s=0$ ($j\neq s$) is obtained by defining $c^s_l = 0$, $c^j_l =  1$ and $c^i_l =  0$ for all $i\neq j$.
\end{proof}

\begin{definition}
A vector $\mathbf{v}\in V$ is called \emph{finitely central} if $\mathbf{v}
$ is $I$-central for some nonempty finite subset $I$ of $E$. We denote by $E[R]$ the set of all finitely central elements of $V_\mathbf{E}$.
\end{definition}

%As a consequence of Theorem \ref{prop:semimodulecentral}, if $V$ admits finitely central elements, then the semiring $R$ has zero-divisors. 
%Also, observe that an $n$-subset of $X$ (Example \ref{exa:partition}) is $n$-central in the Boolean vector space $\mathrm{Set}_{n}(X)$ iff it is an $n$-partition of $X$.

\begin{lemma}
\label{lem:Ebp}

\begin{itemize}
\item[(i)] The set $E[R]$ is a subuniverse of the algebra $(V,g^{V})_{g\in \tau }$.

\item[(ii)] If $E$ has finite cardinality $n$, then $E[R]$ is  closed under the
operation $q^V_{E}$ and the algebra $(E[R],g^{V},q^V_{E},\e^V)_{\e\in E,g\in \tau }$ is an $n\mathrm{BA}$.
\end{itemize}
\end{lemma}

\begin{proof}
(i) Without loss of generality, assume $g\in \tau $ to be a binary operator.
Let $\mathbf{v}$ be $I$-central and $\mathbf{t}$ be $J$-central (for $I,J$
finite subsets of $E$). Then we set: 
\begin{equation}
\mathbf{w}=g^{V}(\mathbf{v},\mathbf{t})=\sum_{\dd,\e\in E}(v_{\dd} t_{\e})g^{\mathbf{E}}(\dd,\e).  \label{eq:g}
\end{equation}
We show that $\mathbf{w}$ is $H$-central, where $H=\{g^{\mathbf{E}}(\dd,\e):\dd\in I,\e\in J\}$. Let 
\[
H_{\cc}=\{(\dd,\e)\in I\times J:g^{\mathbf{E}%
}(\dd,\e)=\cc\}.
\]
Then, by (\ref{eq:g}) $w_{\cc}=\sum_{(\dd,\e)\in H_{\cc}}v_{\dd}t_{\e}$, is the $\cc$-coordinate of $\mathbf{w}$. The conclusion follows from
Theorem \ref{prop:semimodulecentral} by verifying that $w_{\cc}=0$ for $\cc\notin H$, $\sum_{\e\in H}w_{\e}=1$, $w_{\dd} w_{\dd}=w_{\dd}$, $w_{\dd} w_{\e}=0$ for $\dd\neq \e$, and for all $x\in R$, $w_{\e}x=x w_{\e}$. % In the proof we use the corresponding properties of the vectors $\mathbf{v}$ and $\mathbf{t}$. 

(ii) If an element $\mathbf{v}$ is $I$-central, then by Theorem \ref{prop:semimodulecentral} $\mathbf{v}$ is also $J$-central for every finite $J$ such that $I\subseteq J\subseteq E$. Then the
conclusion follows because $E$ is finite and all elements of $E[R]$ are $E$-central, so that $E[R]$ is the set of all $E
$-central elements of the algebra $(V,g^{V},q_{E},\e^{V})_{g\in \tau ,\e\in E}$.
\end{proof}

Let us now introduce the following relevant notion:
\begin{definition}
The algebra $\mathbf{E}[R]=(E[R],g^{\mathbf{E}[R]})_{g\in \tau} $, called the \emph{semiring power} of $\mathbf{E}$ by $R$, is the algebra of finitely central elements of the algebra $V_\mathbf E$.
\end{definition}

%$\mathbf{E}[R]$ is an expansion of the algebra $\mathbf{E}$ because $E\subseteq E[R]$.

If $R$ is a Boolean algebra and $\mathbf E$ an arbitrary algebra, let $\mathcal C(R^*,\mathbf E)$ be the set of continuous functions from the Stone space $R^*$ to $E$, giving $E$ the discrete topology.
$\mathcal C(R^*,\mathbf E)$ is a subuniverse of the algebra $\mathbf E^{R^*}$, and it is called the \emph{Boolean power} of $\mathbf E$ by $R$ \cite{BS}.

\begin{theorem}
\label{prop:Ebp}
 If $R$ is a Boolean algebra, then $\mathbf{E}[R]$ is isomorphic to
 the Boolean power $\mathcal C(R^*,\mathbf E)$.
% is subdirectly embeddable into a power of $\mathbf{E}$.
\end{theorem}

\begin{proof} Let $\mathbf{v} = \sum_{\e_i\in I} v_{\e_i} \e_i\in E[R]$
for some $I\subseteq_{\mathrm{fin}}E$.  Given an ultrafilter $F\in R^*$, by Theorem \ref{prop:semimodulecentral} there exists exactly one $\e_i\in I$ such that  $v_{\e_i}\in F$.
Then we define $f_\mathbf v: R^*\to E$ as follows, for every $F\in R^*$:
$$f_\mathbf v(F)=\ \text{the unique $\e_i\in I$ such that $v_{\e_i}\in F$}.$$
 The map $\mathbf v \mapsto f_\mathbf v$ is an isomorphism from $\mathbf{E}[R]$ to $\mathcal C(R^*,\mathbf E)$.
%we have that, for every ultrafilter $F\in X$,  there exists a unique $\e_i\in I$ such that $v_{\e_i}\in F$.
%Then the clopens $C_i=\{F\in X: v_{\e_i}\in F\}$ ($i=1,\dots,n$) determine an $n$-partition of the Stone space $X$. The set of partitions of $X$ by clopens is in bijective correspondence with the set  of continuous functions from $X$ to the discrete space $E$.
\end{proof}

In the remaining part of this section we prove that every semiring power is isomorphic to a Boolean power.

Let $R$ be a semiring. An element $r\in R$ is said to be
\begin{enumerate}
    \item[(i)] \emph{complemented} if there exists an element $s\in R$ such that $ r+s=1$ and  $rs=0$;
    \item[(ii)] \emph{commuting} if, for all $t\in R$, $rt=tr$.
\end{enumerate}
The complement of an element $r$ is unique, and will be denoted by $r'$. 
Indeed, if $s$ is another complement of $r$ then we have: $r' =r'(r+s)=r's$ and $s=(r+r')s= r's$.

We will write $C(R)$ for the set of complemented and commuting elements of $R$.
Every element of  $\mathrm{C}(R)$ is idempotent, because $r=r(r+r')=r^2+rr'=r^2$.

The following lemma shows that $C(R)$ is in fact a Boolean algebra with respect to the operations $r\lor s = r+r's$, $r\land s=rs$ and the above defined complementation.

\begin{lemma}
$C(R)$ is a Boolean algebra.
\end{lemma}

\begin{proof} %First we show that $C(R)$ is closed under the defined operations.

 First, we show that $C(R)$ is closed under complementation. If $r\in \mathrm{C}(R)$, then we prove that $r'$ is commuting: $r't=r't(r+r')=r'tr+r'tr'=r'rt+r'tr'=0+r'tr'=r'tr'$. By symmetry we also have $tr'=r'tr'$.
 
If $r,s\in \mathrm{C}(R)$, then we prove that $r's'$ is the complement of $r\lor s$: $(r\lor s)+r's' = r+r's + r's' = r+r'(s+s')=r+r'=1$ and $(r\lor s)r's' =(r+r's)r's'= rr's' +r'sr's'=0+0=0$, by commutativity. 
Moreover, it is not difficult to check that $r\lor s$ is commuting.

As regards idempotency, $r\lor r = r+r'r = r+0=r$.

For De Morgan Laws, $(r\lor s)'= r's'$ and $(rs)'=r'+rs'$, since $(r+r's)r's' = rr's'+r'sr's' =0+r'^2ss'=0$ and $r+r's+r's' = r+r'(s+s')=r+r'=1$.

Concerning commutativity, $r\lor s= r+r's = (s+s')(r+r's)= sr+s'r+sr's+s'r's=sr+s'r+sr's= s(r+r')+s'r=s+s'r=s\lor r$.

Finally, $(r\lor s)\lor t= r+r's+(r\lor s)'t= r+r's+r's't= r+r'(s+s't)= r\lor(s\lor t)$, i.e. associativity holds true.

We leave to the reader the verification of the other laws.
\end{proof}

The following lemma presents two useful properties of $C(R)$.
\begin{lemma}\label{lem:somma} Let $a_1,\dots,a_n\in C(R)$.
\begin{itemize}
\item[(a)]  $(a_1+\dots+a_n)(a_1\lor\dots\lor a_n)=a_1+\dots+a_n$.
%\item[(b)] If $a_1+\dots+a_n=1$ then $a_1\lor\dots\lor a_n=1$.
\item[(b)] If $a_ia_j=0$ ($i\neq j$), then $a_1+\dots+a_n= a_1\lor\dots\lor a_n$.
\end{itemize}
\end{lemma}

\begin{proof} (a) 
\begin{align*}
(\sum_{i=1}^{n}a_i)(\bigvee_{i=1}^{n}a_i)&=\sum_{i=1}^{n}(a_i(\bigvee_{i=1}^{n}a_i))\\
&=\sum_{i=1}^{n}(a_i\land(\bigvee_{i=1}^{n}a_i))\\
&=\sum_{i=1}^{n}a_i
\end{align*}

(b) The proof is by induction. 

($n=2$):  $a_1\lor a_2= (a_1\lor a_2)+a_1a_2=a_1+a_1'a_2+a_1a_2=a_1+(a_1+a_1')a_2=a_1+a_2$.
%$a + b = (a + b)(a + a')(b + b')=(a + b)(ab + a'b+ ab' + a'b') =ab+ab' +ab+a'b=ab' +ab+a'b=a(b+b')+a'b=a+a'b=a\lor b$.

($n>2$) 
\begin{align*}
a_1\lor a_2\lor \dots\lor a_n&= a_1\lor(\bigvee_{j=2}^n a_j)\\
&= a_1\lor(\sum_{j=2}^n a_j)\\
&=a_1+(\sum_{j=2}^n a_j)\\
&= a_1+ a_2+ \dots+ a_n
\end{align*}
by applying the basic case of induction to $a_1$ and $\sum_{j=2}^n a_j$.
\end{proof}

\begin{theorem}\label{thm:semiboolean} Let $R$ be a semiring and $\mathbf E$ be an algebra.
The semiring power  $\mathbf E[R]$ is isomorphic to the Boolean power $\mathcal C(C(R)^*,\mathbf E)$.
\end{theorem}

\begin{proof} By Theorem \ref{prop:Ebp} it is sufficient to prove that $\mathbf E[R]$ coincides with $\mathbf E[C(R)]$.

Let $\mathbf v=\sum_{i=1}^n v_i\e_i\in E[R]$.
 By Theorem \ref{prop:semimodulecentral} the coordinates $v_1,\dots, v_n$ are idempotent,  commuting and  orthogonal, i.e., $\sum_{i=1}^nv_i=1$ and $v_iv_j=0$, if $i\neq j$. 
 Since $v_i + \sum_{j\neq i} v_j = 1$ and $v_i(\sum_{j\neq i} v_j)=0$,  then we get  $v_i\in C(R)$, for every $i=1,\dots,n$.  By $v_1+\dots+v_n=1$ and by Lemma \ref{lem:somma}(a) we derive $v_1\lor\dots\lor v_n=1$. Thus $\mathbf v\in E[C(R)]$. 

For the converse, let $\mathbf v=\sum_{i=1}^n v_i\e_i\in E[C(R)]$. Then the coordinates $v_1,\dots, v_n$ are idempotent,  commuting and orthogonal in the Boolean algebra $C(R)$, i.e.,  $\bigvee_{i=1}^nv_i=1$ and $v_i\land v_j=v_iv_j=0$ ($i\neq j$). By Lemma \ref{lem:somma}(b), we derive $v_1+ \dots+ v_n=1$. Thus, $\mathbf v\in E[R]$.

The operations on $\mathbf E[R]$ and those on $\mathbf E[C(R)]$  coincide.
\end{proof}

%For every $I=\{\e_1,\dots,\e_n\}\subseteq E$, we define $\mathbf E[R,I]$ to be the subreduct of $\mathbf E[R]$ of $I$-central elements, which is an $n$BA with respect to the operator $q_I$ and the designated constants in $I$. Let $B_I$ be the Boolean algebra of the coordinates of $\mathbf E[R,I]$ as defined in Definition \ref{def:coo}.
%We now show that $B_I$ and $\mathbf{C}(R)$ are isomorphic and this isomorphism is independent of the choice of the subset $I$ of $E$.

%\begin{corollary}
% Every semiring power  $\mathbf E[R]$ is  a subdirect power of $\mathbf E$.
%\end{corollary}
%
%\begin{proof}
% By Theorem \ref{thm:semiboolean} and the corresponding result for Boolean powers (see \cite[Thm. 5.4 in Ch. 5]{BS}).
%\end{proof}

%To round off this section, we make a note of the fact that, by Theorem 
%\ref{prop:Ebp}(ii), the algebras $R[\mathbf{E}]$ and $\mathbf{E}$
%satisfy the same identities whenever $R$ is a bounded distributive
%lattice.

%\begin{example} \emph{(Linear lambda calculus)} Let $\mathbf E$ be the term algebra of the minimal $\lambda$-theory $\lambda\beta$ and let $R$ be a semiring. The linear lambda calculus is the semiring power $R[\mathbf E]$ of the term algebra of $\lambda\beta$ by $R$.
%\end{example}

\section{Representation Theorems}
In this section we discuss two representation theorems. Actually, we show that any pure $n$BA $\mathbf A$ is isomorphic to the algebra of $n$-central elements of a certain Boolean vector space, namely, the free $n$-generated Boolean vector space over the Boolean algebra $B_\mathbf A$ defined in Section \ref{sec:ba} (Theorem \ref{thm:rep1}).

Moreover, for any finite $n$BA $\mathbf A$, any algebra in the variety $V(\mathbf A)$ is isomorphic to a Boolean power of the generator $\mathbf A$ (Theorem \ref{thm:nBA->E[BA]}). A notable consequence of this result and of Theorem \ref{prop:nbaprim} is Foster's Theorem for primal varieties (cfr. \cite[Thm. 7.4]{BS}).

We will see that the main technical tool will be the concept of an “inner'' Boolean algebra $B_\mathbf A$ contained inside every $n$BA $\mathbf A$. 
We use $B_\mathbf A$ to define the coordinates $a_1,\dots, a_n\in B_\mathbf A$ of every element $a\in A$. Then, the map that associate to every $a\in A$ its coordinates $(a_1,\dots,a_n)$ provides the embedding of $\mathbf A$ into the $n$BA  of $n$-central elements of the Boolean vector space $B_\mathbf A^n$. 

% how to represent ; namely, the free $n$-generated Boolean vector
%space over $B_{\mathbf{A}}$. Due to this representation theorem, we can view
%the Boolean vector space generated by the $n\mathrm{BA}$ of the $n$-central
%elements of a given $n\mathrm{CA}$ $\mathbf{A}$ as a kind of \emph{linear
%approximation} of $\mathbf{A}$. Another notable consequence of our result is
%that Foster's Theorem for primal varieties (cfr. \cite[Thm. 7.4]{BS}),
%according to which any member of a variety generated by a primal algebra is
%a Boolean power of the generator, can be obtained as a corollary.
%

\subsection{The inner Boolean algebra of an $n$BA}\label{sec:ba}

Any $n\mathrm{CA}$ $\mathbf{A}$ accommodates within itself Church algebras of dimension $m$, for any $m<n$. Indeed, if $\mathbf{A}$ is an $n\mathrm{CA}$, and $m<n$, set 
\begin{equation}
p_m(a,a_{1},\dots ,a_{m})=q(a,a_{1},\dots ,a_{m},\e_{m+1},\dots ,\e_{n}).
\label{eq}
\end{equation}%
It is straightforward to verify that $\mathbf{A}$
is an $m\mathrm{CA}$ with respect to the defined $p_m$ and $\e_1,\dots,\e_m$. 
A tedious but straightforward computation shows the following lemma.
 
\begin{lemma}\label{lem:nnn} Let $x$ be an $n$-central element of a pure $n\mathrm{CA}$ $\mathbf A$. Then
$x$ is $m$-central  iff 
$p_m(x,y,y,\dots,y) =y$, for all $y\in A$.
\end{lemma}

Let $\mathbf A$ be a pure $n$BA. The set 
$B_{\mathbf{A}}=\{x\in A\ :\ p_2(x,y,y) =y\}$ of the  $2$-central elements of $\mathbf{A}$
with respect to the ternary term operation $p_2$ and constants $\e_1,\e_2$ is a Boolean algebra (see \cite{first}).

\begin{definition}\label{def:coo} Let $\mathbf{A}$ be a pure $n\mathrm{BA}$. 
The Boolean algebra $B_{\mathbf{A}}$,  whose operations are defined as follows: 
$$
x\wedge y=p_2(x,\e_{1},y);\quad x\vee y=p_2(x,y,\e_{2});\quad
\lnot x=p_2(x,\e_{2},\e_{1});\quad 0=\e_{1};\quad 1=\e_{2},
$$
is called the \emph{Boolean algebra of the coordinates of $\mathbf A$}.
\end{definition}

As a matter of notation, we write $q(x,y,z,\bar{u})$ for $q(x,y,z,u,\dots,u)$.
The next lemma gathers some useful properties of the algebra {$B_{\mathbf{A}}$}.

\begin{lemma}\label{lem:qclosure} Let $\mathbf{A}$ be a pure $n\mathrm{BA}$, $y,x_{1},\dots ,x_{n}\in B_{\mathbf{A}}$ and $a\in A$. Then we have:
\begin{itemize}
\item[(i)] $q(y,\e_{1},\e_{2},\bar{\e}_{1}) = y$.
\item[(ii)] $q(a,x_{1},\dots,x_{n})\in B_{\mathbf{A}}$.
\item[(iii)] $B_{\mathbf{A}}=\{q(a,\e_{1},\e_{2},\bar{\e}_{1})\ :\ a\in A\}$.
\end{itemize}
\end{lemma}

\begin{proof}  (i) 
\[
\begin{array}{llll}
q(y,\e_{1},\e_{2},\bar{\e}_{1})   &  = & q(y,\e_{1},\e_{2},p_2(y,\e_1,\e_1),\dots,p_2(y,\e_1,\e_1))&\text{by $y\in   B_{\mathbf{A}}$}\\
  &  = & q(y,\e_{1},\e_{2},q(y,\e_1,\e_1,\e_3,\dots,\e_n),\dots,q(y,\e_1,\e_1,\e_3,\dots,\e_n))&\text{by Def. of $p_2$}  \\
  & =  &   q(y, \e_{1},\e_{2},\e_3,\dots,\e_n)&\text{by ($\mathrm{B4}$)}\\
  &=& y.&
\end{array}
\]
(ii) Let $a\in A$ and $\hat{x}=x_{1},\dots ,x_{n}\in B_\mathbf A$. As $a$ is $n$-central, then by Lemma \ref{lem:nnn}
$q(a,\hat{x})\in B_{\mathbf{A}}$  iff $p_2(q(a,\hat{x}),z,z)=z$, for any $z\in A$.
\begin{equation*}
\begin{array}{llll}
p_2(q(a,\hat{x}),z,z) & = & {q(q(a,\hat{x}),z,z,\e_{3},\dots ,\e_{n})}&
\\ 
& = & q(a,\dots ,q(x_{i},z,z,\e_{3},\dots ,\e_{n}),\dots )&\text{by (B3)} \\ 
& = & {q(a,\dots ,p_2(x_{i},z,z),\dots )} &\text{by Def. of $p_2$}\\ 
& = & q(a,z,\dots ,z,\dots ,z) &\text{by (B2)}\\ 
& = & z&\text{by (B2)}
\end{array}
\end{equation*}
(iii) By (i)-(ii) applied to $q(a,\e_{1},\e_{2},\bar{\e}_{1})$.
\end{proof}

%\begin{proof} Let $t= q(a,\e_1,\e_{2},\bar{\e}_{1})$. Then we have:
% \begin{equation*}
%\begin{array}{llll}
%t &=&q(a,\e_1,\e_2,p(a,\e_{1},\e_{1}),\dots,p(a,\e_{1},\e_{1}))&\\
%&=& q(a,\e_1,\e_2,\overline{q(a,\e_1,\e_1,\e_3,\dots,\e_n)})&\\
%&=_\text{(B4)}& q(a,\e_1,\e_2,\e_3,\dots,\e_n)&\\
%&=&a.
%\end{array}
%\end{equation*}
%\vspace{-0.3cm}
%\end{proof}

%\begin{lemma}
% Let $\mathbf{A}$ be an $n\mathrm{BA}$. If $x_{1},\dots ,x_{n}\in B_{\mathbf{A}}$ and $a\in A$, then $q(a,x_{1},\dots,x_{n})\in B_{\mathbf{A}}.$  
%\end{lemma}

%Making good use of the previous lemmas, we obtain the following:

%\begin{lemma}
%\label{lemma:strano} For every $a\in B_{\mathbf{A}}$ and every $%
%x,y,z_3,\dots,z_n\in A$ we have that 
%\begin{equation*}
%p(a,x,y) = q(a,x,y,z_3,\dots,z_n). 
%\end{equation*}
%\end{lemma}
%\begin{proof} Let $t= p(a,x,y)$. Then we have:
%\begin{equation*}
%\begin{array}{lll}
%t & = & q(a,x,y,\e_{3},\dots ,\e_{n}) \\ 
%& = & q(q(a,\e_{1},\e_{2},\bar{\e}_{1}),x,y,\e_{3},\dots ,\e_{n}) \\ 
%& = & q(a,x,y,\bar{x}) \\ 
%& = & q(q(a,\e_{1},\e_{2},\bar{\e}_{1}),q(a,\bar{x}),q(a,\bar{y}),q(a,\bar{z}_{3}),\dots ,q(a,\bar{z}_{n}))\\
%& = & q(a,x,y,z_{3},\dots ,z_{n}).%
%\end{array}%
%\end{equation*}
%\vspace{-0.3cm}
%\end{proof}

\subsection{The coordinates of an element}

Let $\mathbf{A}$ be a pure $n\mathrm{BA}$, 
and $\sigma$ be a permutation of $1,\dots,n$. For any $a\in A$, 
we write $a^\sigma$ for $q(a, \e_{\sigma 1},\dots,\e_{\sigma n})$. In particular, if $(2i)$ is the transposition defined by $(2i)(2)=i$, $(2i)(i)=2$ and $(2i)(k)=k$ for $k\neq 2,i$, then  we have 
$$a^{(2i)} =q(a, \e_{(2i) 1},\dots,\e_{(2i) n})= q(a, \e_1,\e_{i}, \e_3,\dots,\e_{i-1},\e_{2},\e_{i+1},\dots,\e_n).$$
We write $q(a, \e_{i}/2;\e_{2}/i)$ for $a^{(2i)}$. 

\begin{definition} Let $(B_{\mathbf{A}})^{n}$ be the Boolean vector space of dimension $n$ over $B_{\mathbf{A}}$. 
 The \emph{vector of the coordinates} of an element $a\in A$ is  a tuple $(a_{1},\dots ,a_{n})\in(B_\mathbf{A})^n$, where 
$$a_{i}=q(a^{(2i)},\e_{1},\e_{2},\bar{\e}_{1}).$$ 
\end{definition}
Observe that by Lemma \ref{lem:qclosure}(ii) $a_{i}\in B_{\mathbf{A}}$ for every $i$. The next lemma shows that the
coordinate $a_{i}$ admits a simpler description.

\begin{lemma}
\label{lem:coordinate} If $\mathbf{A}$ is a pure {$n\mathrm{BA}$} and $a\in A$,
then $a_{i}=q(a,\e_{1},\e_{1},\dots,\e_{1},\e_{2},\e_{1},\dots,\e_{1})$, where $\e_2$ is at position $i$.
\end{lemma}

\begin{proof} 
\begin{equation*}
\begin{array}{llll}
a_{i} & = & q(a^{(2i)},\e_{1},\e_{2},\bar{\e}_{1})& \text{by Def.} \\ 
& = & q(q(a,\e_{i}/2;\e_{2}/i),\e_{1},\e_{2},\bar{\e}_{1}) &\\ 
& = & q(q(a,\e_{1},\e_{i},\e_{3},\dots ,\e_{i-1},\e_{2},\e_{i+1},\dots
,\e_{n}),\e_{1},\e_{2},\bar{\e}_{1})& \\ 
& = & q(a,\e_1,\e_1,\dots,\e_1,\e_2,\e_1,\dots,\e_1)&\text{$\e_2$ at position $i$}\\
& = & q(a,\e_{2}/i;\e_{1}/\bar{\imath}).&
\end{array}%
\end{equation*}
\end{proof}

\begin{example} By Theorem \ref{prop:semimodulecentral} an $n$-subset of $X$ (see Example \ref{exa:parapa}) is $n$-central iff it is an $n$-partition of $X$.
If $ P = (P_1,\dots,P_n)$ is an $n$-partition of $X$, then
the $i$-th coordinate of $P$ is
$(X\setminus P_i, P_i,\emptyset,\dots,\emptyset)$.
\end{example}

The following lemma follows directly from the definition.
\begin{lemma}\label{lem:boolcord}
If $a\in B_{\mathbf{A}}$, then 
$a_{1}=\lnot{a};\quad a_{2}=a;\quad a_{k}=0\ {(3\leq k\leq n)}.$
\end{lemma}

The coordinates of the result of an application of $q$ to elements
of $A$ can be expressed as the result of an application of the Boolean
operations of $B_{\mathbf{A}}$ to the coordinates of the arguments.

\begin{lemma}
\label{lem:coordinate2} Let $\mathbf{A}$ be a pure $n\mathrm{BA}$. For every $a,b^{1},\dots ,b^{n}\in A$ we have:
\begin{itemize}
\item[(i)] The coordinates of $a$ are fully orthogonal, i.e., $\bigvee_{i=1}^n a_i = 1$ and $a_i\land a_k=0$, for every $i\neq k$.
 \item[(ii)] $q(a,b^{1},\dots ,b^{n})_{i} =q(a,{(b^{1})_{i},\dots ,(b^{n})_{i}}) 
=\bigvee_{j=1}^{n}(a_{j}\wedge (b^{j})_{i})$, where $a_{j}$ is the $j$-th coordinate of $a$ and $(b^{j})_{i}$ is the $i$-th coordinate of $b^{j}$. 
\end{itemize}
The join $\bigvee $ and the meet $\wedge $ are
taken in the Boolean algebra $B_{\mathbf{A}}$.
\end{lemma}
\begin{proof}
It suffices to check the previous identities in the $n$-element generator $\mathbf n$ of 
the variety of pure $n\mathrm{BA}$s.
\end{proof}
%The following lemma will be useful for the development of our discourse.

%\begin{lemma} Let $\mathbf{A}$ be an $n\mathrm{BA}$. For every $a\in A$
%we have

%\end{lemma}
%\begin{proof}
%Applying Lemma \ref{lem:coordinate}, a long but unproblematic calculation
% shows that $p(a_{i},\e_{1},a_{k})=0$ and $p(a_{1},\bigvee_{i=2}^{n}a_{i},\e_{2})=1$.
%\end{proof}

\begin{lemma}
\label{lem:samecoordinates} Let $\mathbf{A}$ be a pure $n\mathrm{BA}$. If $a,b\in A$ have the same coordinates, i.e. $a_i=b_i$ for all $i$, then $a=b$.
\end{lemma}

\begin{proof} 
%$\mathbf{A}$ is representable as a subdirect product of the generators in $\mathcal{K}=\{\mathbf{B}\in \mathcal{V}(\mathbf A):|B|=n\}$. 
%Since in the definition of vector of coordinates we only use the $q$ operator and the constants, we can safely restrict ourselves to the pure reduct of $\mathbf{A}$. 
%Every pure $n\mathrm{BA}$ $\mathbf{A}$ is isomorphic to a subdirect power of $\mathbf{n}^{I}$, for some set $I$ 
By Theorem \ref{lem:subirr}   $\mathbf{A}$ is a subalgebra of $\mathbf{n}^{I}$, for an
appropriate $I$. It is routine to verify in $\mathbf{n}^{I}$ that if two
elements have the same coordinates then they coincide, whence the same is
true for $\mathbf{A}$.
\end{proof}

\subsection{The main theorems}

We recall that an algebra $\mathbf A$ is a retract of an algebra $\mathbf B$, and we write $\mathbf A\vartriangleleft\mathbf B$, if there exist two homomorphisms $f:\mathbf A\to \mathbf B$ and $g:\mathbf B\to \mathbf A$ such that $g\circ f=\mathrm{Id}_A$.

\begin{theorem} \label{thm:rep1}  Let $\mathbf{A}$ be a pure $n\mathrm{BA}$, $\mathrm{Ce}(B_\mathbf A^n)$ be the $n\mathrm{BA}$ of $n$-central elements of the Boolean vector space $B_\mathbf A^n$, $\mathbf n[B_\mathbf A]$ (resp. $\mathbf A[B_\mathbf A]$) be the semiring power of $\mathbf n$ (resp. $\mathbf A$) by $B_\mathbf A$. Then we have:
$$\mathbf{A} \cong \mathrm{Ce}(B_\mathbf A^n)\cong \mathbf n[B_\mathbf A] \vartriangleleft \mathbf A[B_\mathbf A].$$
\end{theorem}

\begin{proof}
We prove that  $\mathbf{A} \cong \mathbf n[B_\mathbf A]$.
Define the map $f:\mathbf{A}\rightarrow \mathbf{n}[B_{\mathbf{A}}]$
as follows, for any $a\in A$: 
$$f(a)=\sum_{i=1}^{n}a_{i}\e_{i},$$ 
 where $a_{1},\dots ,a_{n}$ are the coordinates of $a$. The map $f$ is well-defined, because the coordinates of $a$ are fully orthogonal in $B_\mathbf A$.
 Moreover, $f$ is injective by Lemma \ref{lem:samecoordinates} and it preserves the operation $q$:
\begin{equation*}
\begin{array}{llll}
q^{\mathbf{n}[B_\mathbf{A}]}(f(a),f(b^{1}),\ldots ,f(b^{n})) & = & q^{\mathbf{n}[B_\mathbf{A}]}(\sum_{i=1}^{n}a_{i}\e_{i},\sum_{i=1}^{n}b^1_{i}\e_{i},\ldots ,\sum_{i=1}^{n}b^n_{i}\e_{i})& \\ 
& = & a_{1}(\sum_{i=1}^{n}b^1_{i}\e_{i}) + \dots+ a_{n} (\sum_{i=1}^{n}b^n_{i}\e_{i})\\
& = & \sum_{i=1}^{n}(a_{1}\land b^1_{i})\e_{i} + \dots+ \sum_{i=1}^{n}(a_{n} \land b^n_{i})\e_{i}\\
& = &\sum_{i=1}^{n}\left( (a_{1}\land b^1_{i}) \lor \dots\lor (a_{n}\land b^n_{i})\right)\e_i&\\

& = & \sum_{i=1}^{n} q^\mathbf A(a,b^{1},\ldots ,b^{n})_{i} \e_i&\text{by Lem. \ref{lem:coordinate2}(ii)}\\

& = & f(q^\mathbf A(a,b^{1},\ldots ,b^{n})).&
\end{array}
\end{equation*}
Next we prove that $f$ is
surjective. Let $\mathbf a=\sum_{i=1}^{n}a_{i}\e_{i}\in \mathbf{n}[B_{\mathbf{A}}]$, so that $a_1\vee\dots\vee a_n=1$ and $a_i\land a_j=0$ for $i\neq j$. Recall from Section \ref{sec:ba} the definition of the ternary operator $p_2$. We will show that $f(b)=\mathbf a$, where 
$$b=p_2(a_{1},p_2(a_{2},p_2(a_{3},p_2(\dots p_2(a_{n-1},\e_{n},\e_{n-1}))\dots),\e_{3}),\e_{2}),\e_{1}).$$
We write $b^{n-1}=p_2(a_{n-1},\e_{n},\e_{n-1})$ and $b^i=p_2(a_{i},b^{i+1},\e_{i})$ ($i=1,\dots,n-2$)
%$b^i=p_2(a_{i},p_2(a_{i+1},p_2(a_{i+2},p_2(\dots p_2(a_{n-1},\e_{n},\e_{n-1}))\dots),\e_{i+2}),\e_{i+1}),\e_{i})$
 in such a way that $b=b^1$.

\noindent As  $\bigvee_{i=1}^n a_i=1$ and $a_i\land a_j=0$ ($i\neq j$), we have that the complement $\lnot a_i$ of $a_i$ is equal to 
$\bigvee_{j\neq i}a_{j}$. 
For an arbitrary $x\in A$, we have:
$$
\begin{array}{llll}
p_2(a_i,x,\e_i)_k& =  &q(a_i,x,\e_i,\e_3,\dots,\e_n)_k&\\
& =  &q(a_i,x_k,(\e_i)_k,(\e_3)_k,\dots,(\e_n)_k)&\text{by Lemma \ref{lem:coordinate2}(ii)}\\
&=& ((a_i)_1\land x_k) \lor ((a_i)_2\land (\e_i)_k)&\text{by  Lemma \ref{lem:coordinate2}(ii)}\\
&=&(\lnot a_i\land x_k) \lor (a_i\land (\e_i)_k)&\text{by Lemma \ref{lem:boolcord}}\\
\end{array}
$$
Since $\e_1=0$ is the bottom  and $\e_2=1$ is the top of $B_\mathbf A$, then by Lemma \ref{lem:coordinate} we obtain:
$$p_2(a_i,x,\e_i)_k=\begin{cases}(\lnot a_i\land x_i)\lor (a_i\land \e_2) = (\lnot a_i\land x_i) \lor a_i &\text{if $k=i$}\\
(\lnot a_i\land x_k) \lor (a_i\land \e_1) = \lnot a_i\land x_k &\text{if $k\neq i$}\\
 \end{cases}
$$
If the coordinates $x_j$ are equal to $0$ for every $j\leq i$, then we derive:
$$(\forall j\leq i.\ x_j=0)\ \Rightarrow\ p_2(a_i,x,\e_i)_k=\begin{cases} a_i &\text{if $k=i$}\\
0 &\text{if $k< i$}\\
\lnot a_i\land x_k &\text{if $k> i$}\\
 \end{cases}
$$
If $i=n-1$ and $x=\e_n$,  then we have: 
 $$(b^{n-1})_k=p_2(a_{n-1},\e_{n},\e_{n-1})_k=\begin{cases} a_{n-1} &\text{if $k=n-1$}\\
0 &\text{if $k< n-1$}\\
\lnot a_{n-1}\land (\e_{n})_n=\lnot a_{n-1}\land \e_{2}=\lnot a_{n-1} &\text{if $k=n$}\\
 \end{cases}
$$
By iterating we get 
 $$(b^{i})_k=p_2(a_{i},b^{i+1},\e_{i})_k =\begin{cases} a_{k} &\text{if $i\leq k\leq n-1$}\\
0 &\text{if $k< i$}\\
\lnot \bigvee_{j=i}^{n-1} a_j &\text{if $k=n$}\\
 \end{cases}
$$
Hence, recalling that $b=b^1$, we have the conclusion $f(b)=\mathbf a$.

%$$
%\begin{array}{lll}
%\mathbf a  & =  &  a_{1}\e_{1}+\lnot a_1(a_{2}\e_{2}+\lnot a_2(\dots (a_{n-1} \e_{n-1}+\lnot a_{n-1} \e_n)\dots) \\
%  & =  & f(p(a_{1},p(a_{2},p(a_{3},p(\dots p(a_{n-1},\e_{n},\e_{n-1}))\dots),\e_{3}),\e_{2}),\e_{1})).  \\
%\end{array}
%$$
The definition of the isomorphism $\mathrm{Ce}(B_\mathbf A^n)\cong \mathbf n[B_\mathbf A]$ is straightforward. 

Since $\mathbf n[B_\mathbf A]$ is a subalgebra of $\mathbf A[B_\mathbf A]$, then the retraction  $\mathbf n[B_\mathbf A] \vartriangleleft \mathbf A[B_\mathbf A]$ can be defined as follows. By Theorem \ref{lem:subirr} $\mathbf A$ is a subalgebra of $\mathbf{n}^I$ for some set $I$. If $i\in I$, then we denote by $\pi_i: \mathbf A \to \mathbf n$  the projection in the $i$-th component of $\mathbf n^I$. 
Then we define $g:\mathbf A[B_\mathbf A]\to \mathbf n[B_\mathbf A]$ as follows, for every $\sum_{a\in A} v_a a\in \mathbf A[B_\mathbf A]$: $g(\sum_{a\in A} v_a a)=\sum_{a\in A} v_a \pi_i(a)$. It is easy to check that $g$ determines a retraction.
\end{proof}

We are now ready to extend the above result to any similarity type. 

\begin{theorem}
\label{thm:nBA->E[BA]} Let $\mathbf{A}$ be an 
 $n\mathrm{BA}$ of type $\tau$, whose minimal subalgebra $\mathbf E$ has finite cardinality $n$. Then we have: 
 $$\mathbf A\cong \mathbf{E}[B_\mathbf{A}]\cong \mathcal C(B_\mathbf A^*, \mathbf E).$$
 % is isomorphic to the Boolean semiring $\mathbf{E}[B_\mathbf{A}]$. 
\end{theorem}

\begin{proof} As the algebra $\mathbf E$ is an expansion of the $n$BA $\mathbf n$ by operations of type $\tau$, it is sufficient to prove that the map $f:\mathbf{A}\rightarrow \mathbf{n}[B_{\mathbf{A}}]$, defined in the proof of Theorem \ref{thm:rep1}, is a homomorphism for the type $\tau$. The other isomorphism follows from Theorem \ref{prop:Ebp}.

Let now $g\in \tau$ be an operator that is
supposed to be binary to avoid unnecessary notational issues. 
We observe that, for $a,b\in A$,
\[
 g^{\mathbf{E}[B_\mathbf{A}]}(f(a),f(b))   =  g^{\mathbf{E}[B_\mathbf{A}]}(\sum_{i=1}^{n}a_{i}\e_{i},\sum_{j=1}^{n}b_{j}\e_{j})  
    =   \sum_{i=1}^{n}\sum_{j=1}^{n} (a_{i}\land b_{j})g^{\mathbf{E}}(\e_{i},\e_{j}) =\sum_{k=1}^{n}w_{k}\e_{k},
\]
where 
$$w_{k}=\bigvee_{\{(i,j):g^{^{\mathbf{E}}}(\e_{i},\e_{j})=\e_{k}\}} a_{i}\land b_{j}.$$
We have to show that $f(g^\mathbf A(a,b))$ is equal to
$\sum_{k=1}^{n}w_{k}\e_{k}$, that is, $g^\mathbf A(a,b)_k = w_k$ for every $1\leq k\leq n$.
Since by hypothesis the variety $V(\mathbf A)$ generated by $\mathbf A$ coincides with the variety $V(\mathbf E)$ generated by $\mathbf E$, then we get the conclusion if the variety $V(\mathbf E)$ satisfies the following identities:
$$g(x,y)_k= g(x_k,y_k)= \bigvee_{\{(i,j):g^{^{\mathbf{E}}}(\e_{i},\e_{j})=\e_{k}\}} x_{i}\land y_{j},$$
where $g(x,y)_k$, $x_i,y_j, x_k,y_k$ are coordinates. 
The first identity follows, because $g^{\mathbf{A}}(x,y)_{k}=q(g^\mathbf{A}(x,y),\e_{1},\dots,\e_1,\e_2,\e_1,\dots,\e_1)=_{(B3)}g^\mathbf{A}(x_k,y_k)$, where $\e_2$ is at $k$-position. The second identity can be easily checked in generator $\mathbf E$ by considering that $B_\mathbf E=\{\e_1,\e_2\}$.
\end{proof}

One of the most remarkable properties of the $2$-element Boolean algebra, called \emph{primality} in universal algebra \cite[ Sec. 7 in Chap. IV]{BS}, is the definability of all finite Boolean functions in terms of a certain set of term operations, e.g. the connectives {\sc and, or, not}. This property is inherited by $n$BAs. 

\begin{definition}
\begin{enumerate}
    \item Let $\mathbf{A}$ be a nontrivial $\tau $-algebra. $\mathbf{A}$ is \emph{primal} if it is of finite cardinality and, for every function $f:A^{n}\rightarrow A$ ($n\geq 0$), there is a $\tau $-term $t(x_1,\dots,x_n)$ such that for all $a_{1},\dots,a_{n}\in A$, $f( a_{1},\dots,a_{n}) =t^{\mathbf{A}}(a_{1},\dots,a_{n})$.
    \item A variety $\mathcal V$ is primal if $\mathcal V= V(\mathbf A)$ for a  primal algebra $\mathbf A$.
\end{enumerate}
\end{definition}

%In 1969, Hu showed that primal algebras are like Boolean algebras in many
%respects. {Hu's duality \cite{Hu1, Hu2}, in fact,} generalises Stone's
%duality for Boolean algebras to the effect that the concrete category that
%corresponds to the variety generated by a primal algebra is dually
%equivalent to the category of Boolean spaces --- hence, equivalent to the
%concrete category of Boolean algebras. More than that, a variety is
%categorically equivalent to the variety of Boolean algebras iff it is the
%variety generated by a primal algebra.

%\begin{definition}\label{def:hnf}
% A $\nu_n$-term is a \emph{head normal form} (hnf, for short) if it is defined according to the following grammar:
%$t,t_i ::= \e_i\ |\ x\ |\ q(x,t_1,\dots,t_n)$,
%where $x$ is an arbitrary variable. The occurrence of the variable $x$ in the hnf $t\equiv q(x,t_1,\dots,t_n)$ is called \emph{head occurrence} of $x$ into $t$.
%\end{definition}

%\begin{lemma}\label{lem:easy} Let $\mathbf{A}$ be a finite $n\mathrm{BA}$ of cardinality $n$. Then, for every function $f:A^{k}\rightarrow A$, there exists a canonical hnf $t$ such that $f=t^\mathbf A$.
%\end{lemma}

It was shown in \cite{SBLP} that:

\begin{theorem}\label{prop:nbaprim} Let $\mathbf{A}$ be a finite $\tau$-algebra of cardinality $n$. Then $\mathbf A$ is primal if and only if it is an $n\mathrm{BA}$. 
\end{theorem}

%\begin{theorem}
%If $\mathbf{P}$ is a primal algebra of cardinality $n$, then the variety generated by $\mathbf{P}$ is a variety of $n\mathrm{BA}$s.
%\end{theorem}
%\begin{proof} If $P=\{p_1,\dots,p_n\}$, then the $q$-operator satisfying
%$q(p_i,x_1,\dots,x_n) = x_i$, for $i=1,\dots, n$,
%is term-definable.
%\end{proof}

It follows that,
if $\mathbf{A}$ is a primal algebra of cardinality $n$, then the variety generated by $\mathbf{A}$ is a variety of $n\mathrm{BA}$s.
Notice that varieties of $n$BAs generated by more than one algebra are not primal.

As a corollary to Theorem \ref{thm:nBA->E[BA]} and Theorem \ref{prop:nbaprim}, we obtain
Foster's Theorem for primal algebras:

\begin{corollary}
\label{thm:fstr} If $\mathbf{P}$ is a primal algebra of cardinality $n$, then any $\mathbf{A}\in V(\mathbf{P})$ is isomorphic to the Boolean power $\mathcal C(B_\mathbf A^*, \mathbf P)$,
for the Boolean algebra $B_\mathbf A$ defined in Section \ref{sec:ba}.
\end{corollary}

\begin{proof} By Theorem \ref{prop:nbaprim} $\mathbf{P}$ is an $n$BA.  
%Let $P = \{\e_1,\dots,\e_n\}$. As $\mathbf{P}$ is primal, then there exists an $n+1$-ary term operation $t$ such that $$t^\mathbf{P}(\e_i,x_1,\dots,x_n)=x_i,\quad\text{for all $1\leq i\leq n$}.$$
%Then $\mathbf{P}$ is an $n\mathrm{BA}$ of cardinality $n$. 
If $\mathbf{A}
\in V(\mathbf{P})$, then the minimal subalgebra of $\mathbf A$ coincides with $\mathbf P$ itself, because the constants $\e_1,\dots,\e_n$ belongs to $A$ and are closed under the operations of the algebra.  By Theorem \ref{thm:nBA->E[BA]}  $\mathbf A$ is isomorphic to $\mathcal C(B_\mathbf A^*, \mathbf P)$.
\end{proof}

\end{document}